\documentclass[runningheads]{llncs}
\usepackage[T1]{fontenc}
\usepackage{graphicx}
\usepackage[a4paper,margin=3cm]{geometry} 

\usepackage[pdftex,colorlinks,linkcolor=black,urlcolor=black,citecolor=black]{hyperref}
\usepackage{tikz}
\usepackage{comment}

\ifpdf
\fi

\usepackage{amsthm}

\usepackage{amsmath,amssymb}
\definecolor{ForestGreen}{RGB}{34,139,34}
\usepackage{tcolorbox}

\spnewtheorem{observation}{Observation}{\bfseries}{\itshape}

\usepackage{tikz}
\usetikzlibrary{shapes}
\usetikzlibrary{positioning,calc} 
\usepackage{lipsum}

\newcommand{\ta}{\text{\normalfont tree-$\alpha$}}
\newcommand{\pa}{\text{\normalfont path-$\alpha$}}
\newcommand{\pw}{\text{\normalfont pw}}
\newcommand{\bpw}{\text{\normalfont bpw}}
\newcommand{\btw}{\text{\normalfont btw}}
\newcommand{\tw}{\text{\normalfont tw}}

\begin{document}
\title{Subexponential and Parameterized Mixing Times of Glauber Dynamics on Independent Sets}
\titlerunning{Subexponential and Parameterized Mixing Times of Glauber Dynamics}
%
\author{Malory Marin}
\authorrunning{M. Marin}
%
\institute{\small ENS de Lyon, CNRS, Université Claude Bernard Lyon 1, LIP, UMR 5668, \\69342 Lyon Cedex 07, France}
%
\maketitle              

\begin{abstract}
Given a graph $G$, the hard-core model defines a probability distribution over its independent sets, assigning to each set of size $k$ a probability of $\frac{\lambda^k}{Z}$, where $\lambda>0$ is a parameter known as the \emph{fugacity} and $Z$ is a normalization constant. The Glauber dynamics is a simple Markov chain that converges to this distribution and enables efficient sampling. Its \emph{mixing time}—the number of steps needed to approach the stationary distribution—has been widely studied across various graph classes, with most previous work emphasizing the dichotomy between polynomial and exponential mixing times, with a particular focus on sparse classes of graphs.

Inspired by the modern fine-grained approach to computational complexity, we investigate subexponential mixing times of the Glauber dynamics on geometric intersection graphs, such as disk graphs. We further study parameterized mixing times with respect to two structural parameters that can remain small even in dense graphs: the tree-independence number and the path-independence number. Building on a result of Dyer, Greenhill, and Müller, we show that Glauber dynamics mixes in polynomial time on graphs of bounded path-independence number, and in quasi-polynomial time when the tree-independence number is bounded. Moreover, we prove that both bounds are tight via a conductance argument, thereby resolving a question raised in their work.

 This work provides a simple and efficient algorithm for sampling from the hard-core model. Unlike classical approaches that rely explicitly on geometric representations or on constructing decompositions such as tree decompositions or separator trees, our analysis only requires their existence to establish mixing time bounds-these structures are not used directly by the algorithm itself.
\end{abstract}

\section{Introduction}

\subsection{Glauber Dynamics}

An \emph{independent set} in a graph $G$ is a set of pairwise non-adjacent vertices. Determining the size of a maximum independent set in $G$, known as the \emph{independence number} $\alpha(G)$, is one of the most fundamental problems in algorithmic graph theory. The hard-core model is a probability distribution over the set of independent sets of a graph $G=(V,E)$, parameterized by the \emph{fugacity} $\lambda>0$, where the probability of an independent set $I\subseteq V$ is given by
$$
\pi_{G,\lambda}(I) = \frac{\lambda^{|I|}}{Z_G(\lambda)}
$$
where $Z_G(\lambda) = \sum_{J \in \mathcal{I}(G)} \lambda^{|J|}$ is the \emph{partition function}, summing over $\mathcal{I}(G)$, the set of all independent sets of $G$. The \emph{Glauber dynamics} on independent sets of a graph is a Markov chain whose stationary distribution corresponds exactly to the hard-core model. In this process, one starts with an arbitrary independent set and then repeatedly selects a vertex $v\in V$ uniformly at random. If $v$ is in the current independent set, it is removed with probability $\frac{1}{\lambda+1}$; otherwise, it is added with probability $\frac{\lambda}{\lambda+1}$, provided that adding it maintains independence. 

The \emph{mixing time} of the Glauber dynamics refers to the number of steps required for the Markov chain to converge close to its stationary distribution $\pi_{G,\lambda}$. More formally, given two distributions $\mu$ and $\nu$ on $\mathcal{I}(G)$, the \emph{total variation distance} $\Delta(\mu,\nu)$ is defined as $\Delta(\mu,\nu)= \frac{1}{2}\sum_{I\in \mathcal{I}(G)} |\mu(I) - \nu(I)|$. Given two independent sets $I,J\in \mathcal{I}(G)$ and an integer $t\geqslant 0$, let $\mu^t_I(J)$denote the probability that the current independent set is $J$ after $t$ transitions, starting from the independent set $I$. The mixing time $\tau_{G,\lambda}$ is the smallest $t$ such that $\max_{I\in \mathcal{I}(G)}\Delta(\mu^t_I, \pi_{G,\lambda}) \leqslant \frac{1}{4}$. The constant $\frac{1}{4}$ is chosen arbitrarily in $(0,\frac{1}{2})$, and running the Glauber dynamics during $\ell \tau_{G,\lambda}$ steps would guarantee a total variation distance of $2^{-\ell}$. A key question in the study of this chain is whether it exhibits \emph{fast mixing}, meaning the mixing time is at most polynomial in the size of the graph. Recent results have established fast mixing under various conditions, such as when the graph has maximum degree $d$ and the activity parameter $\lambda$ is below a critical threshold $\lambda_c(d)$\cite{chen2021optimal}, while it can be exponential when $\lambda$ is above $\lambda_c(d)$ \cite{mossel2009hardness}. More generally, Sly \cite{sly2010computational} proved that unless NP = RP, no polynomial-time approximation scheme exists for the partition function when $\lambda_c(d) < \lambda < \lambda_c(d) + \varepsilon(d)$. Beyond bounded-degree graphs, Eppstein and Frishberg \cite{eppstein2023rapid} demonstrated that Glauber dynamics mixes in time $n^{O(\tw)}$ for graphs of treewidth $\tw$. Similarly, Bezakova and Sun \cite{bezakova2020mixing} showed that in chordal graphs with separator size $b$, the mixing time is $n^{O(\ln b)}$ for all $\lambda > 0$. Recently, Jerrum \cite{jerrum2024glauber} established a dichotomy in $H$-free graphs, identifying structural conditions under which Glauber dynamics mixes efficiently for all $\lambda > 0$. 

The three results mentioned above share a key feature: they hold for arbitrarily large fugacity $\lambda > 0$. However, they all focus on sparse graphs since in some dense graphs—such as complete bipartite graphs—Glauber dynamics typically mixes in exponential time (see Section~\ref{sec:LowerBounds}). Nevertheless, there are exceptions: for instance, in complete graphs, Glauber dynamics exhibits rapid mixing. \\

This suggests that fast mixing may still be achievable in certain dense graph classes, raising the following question, which we aim to explore:

\begin{center} \begin{tcolorbox}[colback=purple!10, colframe=purple, arc=5pt, boxrule=1pt]{Does Glauber dynamics mix "reasonably" fast on some classes of dense graphs, for arbitrarily large fugacity?} \end{tcolorbox} \end{center}

In 2020, Dyer, Greenhill, and Müller~\cite{dyer2021counting} initiated the study of this question by identifying a class of dense graphs on which Glauber dynamics mixes rapidly: the graphs of bounded bipartite pathwidth. This class, defined as the maximum pathwidth of a bipartite subgraph of a graph, notably includes all claw-free graphs and hence all line graphs. Their work suggests that the modern fine-grained perspective on complexity and algorithms, built upon graph parameters, provides a fruitful framework for characterizing graph classes in which Glauber dynamics exhibits fast mixing. In the following paragraph, we review results on subexponential and parameterized algorithms, which serve as the necessary background and motivation for our contribution.

\subsection{Subexponential and Parameterized Algorithms}

In general, most NP-hard problems require at least $2^{\Omega(n)}$ time on $n$-vertex graphs. However, for certain graph classes, such as planar graphs, many of these problems—including finding a maximum independent set—can be solved in time $2^{O(\sqrt{n})}$. Such a complexity of $2^{o(n)}$ is called \emph{subexponential}. Other graph classes, such as $H$-minor-free graphs and unit disk graphs, also admit similar algorithms. These classes share a key property: they have \emph{separator theorems}.  

For instance, the celebrated planar separator theorem states that any $n$-vertex planar graph contains a balanced separator of size $O(\sqrt{n})$~\cite{lipton1979separator}. This line of research has been extended to more general geometric intersection graphs, which, unlike planar graphs, can be dense. For example, while planar graphs have bounded clique size, the intersection graph of unit disks can have arbitrarily large cliques. More formally, let $\mathcal{X}$ be a collection of $n$ subsets of $\mathbb{R}^d$ ($d\geqslant 2$). The intersection graph $G[\mathcal{X}]$ of $\mathcal{X}$ is the $n$-vertex graph with one vertex per subset of $\mathcal{X}$, and an edge between two vertices if and only if both corresponding subsets have a nonempty intersection. When all subsets are balls in $\mathbb{R}^d$, $G[\mathcal{X}]$ is called a \emph{$d$-dimensional ball graph}. For the special case of $d=2$, we say that $G[\mathcal{X}]$ is a \emph{disk graph}.

Early work established subexponential algorithms with running times of $O(n^{n^{1-\frac{1}{d}}})$, and even $2^{O(n^{1-\frac{1}{d}})}$ for $d$-dimensional geometric graphs~\cite{marx2014limited,smith1998geometric}. More recently, a robust framework was developed to obtain tight running times for geometric graphs~\cite{de2018framework}. This framework is based on a new type of separator, called \emph{clique-based separators}, which has received further attention in~\cite{de2023clique}.  

The notion of a separator is closely linked to the concept of tree decompositions. A \emph{tree decomposition} of a graph $G = (V, E)$ is a pair $\mathcal{T} = (T, \{X_t\}_{t \in V(T)})$, where $T$ is a tree and $\{X_t\}_{t \in V(T)}$ is a collection of subsets of $V$ (called \emph{bags}) satisfying the following properties:
\begin{enumerate}
    \item $\bigcup_{t \in V(T)} X_t = V$, meaning every vertex of $G$ appears in at least one bag.
    \item For every edge $(u,v) \in E$, there exists a node $t \in V(T)$ such that $\{u,v\} \subseteq X_t$.
    \item For each vertex $v \in V$, the set of nodes $\{t \in V(T) \mid v \in X_t\}$ induces a connected subtree of $T$.
\end{enumerate}

The \emph{width} of a tree decomposition is $\max_{t \in V(T)} |X_t| - 1$, and the \emph{treewidth} of $G$, denoted $\tw(G)$, is the minimum width over all possible tree decompositions of $G$. A well-known algorithmic result states that a maximum independent set can be found in time $2^{O(\tw(G))} \cdot n$.  

Such complexities of the form $f(k) \cdot n^{O(1)}$, where $k$ is some parameter and $f$ is a computable function, are called \emph{fixed-parameter tractable} (FPT). A direct consequence of the planar separator theorem is that any planar graph has treewidth $O(\sqrt{n})$, leading to a $2^{O(\sqrt{n})}$ algorithm for finding a maximum independent set in planar graphs.  

However, in dense geometric graphs, the presence of large cliques implies high treewidth, making this parameter unsuitable. To address this, a new parameter that allows for some density was introduced in~\cite{dallard2024treewidth}. The \emph{independence number} of a tree decomposition is defined as $\max_{t \in V(T)} \alpha(G[X_t])$, where $\alpha(G[X_t])$ is the independence number of the subgraph induced by a bag. The \emph{tree independence number}, denoted $\ta(G)$, is the minimum independence number over all possible tree decompositions of $G$. This parameter leads to alternative algorithmic results: a maximum independent set can be found in time $n^{\ta(G)}$. While this is not FPT, it remains polynomial, providing a useful trade-off for dense graphs where treewidth-based approaches are inefficient.

A \emph{path decomposition} is a special case where $T$ is a path. For simplicity, we define a path decomposition only as the ordered set of bags $\{X_i\}_{1\leqslant i \leqslant p}$. The \emph{pathwidth} of $G$, denoted $\pw(G)$, is the minimum width of any path decomposition of $G$. Similarly, the \emph{path independence number} of $G$, denoted $\pa(G)$, is the minimum independence number over all possible path decompositions of $G$.

The \emph{bipartite pathwidth} of a graph $G$, introduced by Dyer, Greenhill, and Müller~\cite{dyer2021counting}, denoted $\bpw(G)$, is defined as the maximum pathwidth of a bipartite induced subgraph of $G$. This parameter was specifically introduced to study the mixing time of Glauber dynamics, where the following result was established (up to some reformulation):

\begin{theorem}[\cite{dyer2021counting}]\label{thm:Dyer}
Let $G$ be an $n$-vertex graph. The mixing time of the Glauber dynamics with fugacity $\lambda > e/n$ satisfies
$$
\tau_{G,\lambda} = (2\tilde{\lambda}n)^{O(\bpw(G))},
$$
where $\tilde{\lambda} = e^{|\ln \lambda|}$.
\end{theorem}

From now on, we will write $\tilde{\lambda}$ to denote $e^{|\ln(\lambda)|}$. In the preliminaries, we will observe that $\bpw(G) \leqslant 2\pa(G)$, which immediately yields a polynomial bound on the mixing time for graphs of bounded path-independence number. Likewise, since every $n$-vertex graph satisfies $\tw(G) = O( \ln(n))\pw(G)$~\cite{korach1993tree}, it follows that the mixing time is quasi-polynomial whenever the \emph{bipartite treewidth} of $G$, defined analogously as the maximum treewidth of a bipartite induced subgraph, is bounded. Moreover, as $\btw(G) \leqslant 2\ta(G)$, the same quasi-polynomial bound applies when the tree-independence number is bounded.

The polynomial bound for bounded bipartite pathwidth is tight by a simple conductance argument (see Section~\ref{sec:LowerBounds}). However, the tightness of the quasi-polynomial bound was left open by Dyer, Greenhill, and Müller.

\subsection{A Motivation from Wi-Fi Network}

The hard-core model plays a crucial role in analyzing the performance of CSMA/CA (Carrier Sense Multiple Access / Collision Avoidance) in wireless networks. CSMA/CA is a listen-before-talk mechanism designed to prevent interference that could result in transmission loss. In this protocol, each access point (AP) listens to its radio channel before transmitting and only transmits if no other AP is detected as transmitting.

A wireless network can naturally be modeled as a graph: vertices represent APs, and two vertices are adjacent if their corresponding APs can detect each other. Under the CSMA/CA protocol, at any given time, the set of transmitting APs forms an independent set in this \emph{conflict graph}, a central tool for evaluating network performance.

In \cite{laufer2015capacity}, the first formal analysis of the global behavior of the network depending on its topology was conducted. In particular, it was shown that at equilibrium, the probability that the current set of transmitting vertices forms a given independent set $I \in \mathcal{I}(G)$ is given by  
$$
\frac{\theta^{|I|}}{Z_G(\theta)}
$$
where $\theta$ is the ratio between the transmission and listening phase durations. In Wi-Fi networks, $\theta$ typically ranges from 20 to 100. Observe that this distribution is exactly  the hard-core model with fugacity $\theta$. 

A key performance metric is the \emph{throughput} of a vertex, defined as the number of bits per second that a vertex can transmit. Neglecting transmission errors and network protocol headers, the throughput of a vertex is proportional to the fraction of time it is transmitting. For $v \in V(G)$, this fraction, denoted by $p_v$, is given by  
$$
p_v = \frac{\sum_{I\in \mathcal{I}(G), v\in I}\theta^{|I|}}{Z_G(\theta)}
$$
Thus, computing $p_v$ for all $v\in V(G)$ is crucial for understanding network behavior. There are two main approaches to compute these values:

\begin{itemize}
    \item \textbf{Exact Computation:} This approach requires an algorithm that counts the number of independent sets of each size in a given graph. Using such an algorithm, one can compute $Z_G(\theta)$ and, for each $v \in V(G)$, the value $Z_{G-N[v]}(\theta)$, where $G-N[v]$ is the graph obtained by removing $v$ and all its neighbors from $G$. Then, use the formula
    $$
    p_v = \frac{\theta Z_{G-N[v]}(\theta)}{Z_G(\theta)}
    $$
    to compute $p_v$ for all $v\in V(G)$.
    \item \textbf{Sampling-Based Estimation:} This approach relies on a probabilistic algorithm that samples independent sets from the hard-core model with fugacity $\theta$. To estimate $p_v$ for some $v \in V(G)$, one can generate multiple independent sets using this sampler and compute the proportion of time $v$ appears in these sets. This value converges almost surely to $p_v$ by the law of large numbers.
\end{itemize}

Both approaches share a common drawback: the underlying problems are hard. More precisely, computing the exact number of independent sets of each size is NP-hard, as it trivially reduces to the problem of finding a maximum independent set. Furthermore, Dyer, Frieze, and Jerrum \cite{dyer2002counting} proved that, unless $\text{RP} = \text{NP}$, no fully polynomial randomized approximation scheme (FPRAS) can exist for counting independent sets (and thus for computing the partition function).
 
However, decades of research in algorithmics and complexity theory have shown that this dichotomy (P vs NP) does not always fully capture the difficulty of a problem, as demonstrated by the development of subexponential and parameterized algorithms. This is particularly true when the considered graphs exhibit specific structural properties. For example, a graph may, in some sense, resemble a simpler graph, such as a path or a tree, in which polynomial algorithm exist. Another important case is when the graph admits a geometric interpretation, which is particularly relevant for wireless networks. Indeed, the conflict graph of a wireless network is often modeled as the intersection graph of disks (or balls) in $\mathbb{R}^2$ (or $\mathbb{R}^3$).

\subsection{Contributions and organization}

In the next section, we introduce the main definitions and tools required for this work.

In Section~\ref{sec:GeometricGraphs}, we establish our first main result: a subexponential mixing time for intersection graphs of balls in $d$-dimensional space.

\begin{theorem}\label{thm:BallGraphs}
Let $G$ be an intersection graph of $n$ balls in $\mathbb{R}^d$ with $d\geqslant 2$. The mixing time of the Glauber dynamics with fugacity $\lambda> 0$ satisfies 
$$
\tau_{G,\lambda} = (2\tilde{\lambda})^{O(n^{1-\frac{1}{d}})}.
$$
\end{theorem}
Note that for such graphs, computing the partition function $Z_G(\lambda)$ exactly, or determining the activation probabilities $p_v$ for all $v \in V(G)$, cannot be done in time $2^{o(n^{1-1/d})}$ under the ETH, by a simple reduction to the problem of computing the independence number $\alpha(G)$ within this class~\cite{de2018framework}. However, by using Glauber dynamics, we obtain a randomized algorithm to estimate the $p_v$s that is not only simple to implement but also matches the complexity of the best possible exact deterministic algorithm.

In fact, we prove a slightly more general result for hereditary graph classes that admit clique-based separators, which include intersection graphs of balls. Notably, this result also extends to intersection graphs of similarly-sized fat objects in $d$-dimensional space, and many other classes of graphs.

In Section~\ref{sec:Parameterized}, we investigate mixing times in graphs of bounded path independence number and tree independence number. We prove the following

\begin{theorem}\label{thm:PathTree}
Let $G$ be a graph with $n$ vertices. The mixing time of the Glauber dynamics with fugacity $\lambda> 0$ satisfies 
\begin{itemize}
    \item $\tau_{G,\lambda} = (b\tilde{\lambda})^{O(\pa(G))} \cdot n^{O(1)}$ where $b$ is the minimum width (plus one) of a path decomposition of path independence number $\pa(G)$ ;
    \item $\tau_{G,\lambda} = n^{O(\ta(G))\cdot \ln(b\tilde{\lambda})}$ where $b$ is the minimum width (plus one) of a tree decomposition of tree independence number $\ta(G)$.
\end{itemize}
\end{theorem}
This theorem can be seen as a refinement of Theorem~\ref{thm:Dyer} in the context of path and tree-independence number since it takes into account the size of the bags in the bounds, while such a parameter do not exist in the context of bipartite pathwidth and treewidth. In  particular, Theorem~\ref{thm:PathTree} generalizes the polynomial mixing time result for chordal graphs with bounded separators ~\cite{bezakova2020mixing}. In addition, with minor modifications to the proof, we also obtain a result for pathwidth, establishing the first FPT mixing time bound for Glauber dynamics in the hard-core model. Bordewich et al.~\cite{bordewich2014subset} previously proved an FPT mixing time bound in terms of pathwidth for a different Glauber dynamics, which they used to establish polynomial mixing time when treewidth is bounded. Our result on bounded pathwidth similarly implies polynomial mixing time for bounded treewidth, offering an alternative proof of the result by Eppstein and Frishberg~\cite{eppstein2023rapid}. Note that, in full generality (that is, when the bag size $b$ may be as large as the number of vertices), Theorem~\ref{thm:PathTree} follows from Theorem~\ref{thm:Dyer} by simple inequalities between the parameters.

Finally, in Section~\ref{sec:LowerBounds}, we establish lower bounds for all our results, relying on conductance arguments. Previous works proved that Theorem~\ref{thm:BallGraphs} is tight in dimension two~\cite{randall2006slow}, and not improvable by far in dimension at least $3$~\cite{borgs1999torpid}. The matching lower bounds for pathwidth and path independence number of Theorem~\ref{thm:PathTree} are straightforward to obtain, as they involve well-known families of graphs, specifically complete bipartite graphs and slight variations thereof. Then, we prove a quasi-polynomial matching lower bound for the case of tree-independence number, which is more intricate and reveals connections to \emph{TAR-reconfiguration}, a well-studied variant of independent set reconfiguration problems. More precisely, we prove the following.

\begin{theorem}\label{thm:LowerBoundTreeAlpha}
For any $t\geqslant 1$, there exist infinitely many graphs $G$ such that $\ta(G)=2t$ and for any $\lambda\geqslant 1$,
$$
\tau_{G,\lambda} \geqslant (\lambda n)^{ct\ln(n)} 
$$
where $n=|V(G)|$ and $c$ is a constant.
\end{theorem}

Since graphs of bounded tree-independence number have bounded bipartite treewidth, it answers a question of Dyer, Greenhill and Muller.

\section{Preliminaries}\label{sec:Preliminaries}

Most of the definitions used in this paper are standard; however, we recall the most frequently used ones for clarity. Additionally, we provide a brief overview of the key techniques used to bound the mixing time of Glauber dynamics. For readers unfamiliar with these concepts, a more detailed introduction on Markov chains can be found in Appendix~\ref{appendix:Markov}.

\subsection{Graphs, Parameters and Partition Function} 

\paragraph{Notations.} Let $G$ be a simple graph. We denote by $V(G)$ and $E(G)$ the set of vertices and the set of edges of $G$, respectively. When there is no ambiguity, we denote by $n$ the number of vertices of $G$, and by $m$ the number of edges of $G$. 

Given a set $R \subseteq V(G)$, we use $G[R]$ to denote the subgraph induced by $R$, and $G - R$ to denote the graph induced by $V(G) \setminus R$. For a vertex $v \in V(G)$, we denote by $N(v)$ the \emph{open neighborhood} of $v$, that is, $N(v) = \{u \in V(G) \mid uv \in E(G)\}$, and by $N[v]$ its \emph{closed neighborhood}, defined as $N[v] = N(v) \cup \{v\}$. Given two sets $A$ and $B$, we denote by $A\Delta B$ the \emph{symmetric difference} of $A$ and $B$, that is $(A\setminus B) \cup (B\setminus A)$.

For a subset $X\subseteq V(G)$, we write $Z_X(\lambda)$ instead of $Z_{G[X]}(\lambda)$ when there is no ambiguity. Given a graph $G$, we define $\mathcal{G}(G)$ as the \emph{reconfiguration graph} of independent sets in $G$. Its vertex set is $\mathcal{I}(G)$, and two independent sets $I, J \in \mathcal{I}(G)$ are adjacent if and only if $|I \Delta J| \leqslant 1$. This graph serves as the support graph for the Glauber dynamics.

\paragraph{Relation between parameters.}Notice that bipartite pathwidth and bipartite treewidth can be related to path- and tree-independence numbers in a straightforward way.

\begin{proposition}\label{prop:BoundBipartite}
For any graph $G$, the inequalities $\bpw(G) \leqslant 2\pa(G)+1$ and $\btw(G) \leqslant 2\ta(G)+1$ hold.
\end{proposition}

\begin{proof}
Suppose that $G$ has a path decomposition $(X_t)_{1\leqslant t\leqslant p}$ of independence number $k=\pa(G)$.  
Let $V'\subseteq V(G)$ be such that $G[V']$ is bipartite. Then $(X_t \cap V')_{1\leqslant t\leqslant p}$ is a path decomposition of $G[V']$.  
Moreover, for every $1\leqslant t \leqslant p$, $G[X_t\cap V']$ is bipartite, and hence $X_t\cap V'$ can be partitioned into two independent sets.  
It follows that 
\[
|X_t\cap V'| \leqslant 2\alpha(G[X_t\cap V']) \leqslant 2k,
\] 
so $\bpw(G)\leqslant 2\pa(G)+1$.  
The argument for $\btw(G)\leqslant 2\ta(G)+1$ is identical, using a tree decomposition instead of a path decomposition.
\end{proof}

Since any bipartite induced subgraph of a claw-free graph is a disjoint union of paths and cycles~\cite{dyer2021counting}, the class of claw-free graphs has bipartite pathwidth at most $2$.  
However, this class has unbounded tree-independence number.  
Indeed, line graphs of grids are claw-free, and their tree-independence number is unbounded, since they have bounded maximum degree but unbounded treewidth.

\paragraph{A technical lemma.} In the next lemma, we present three properties  which are simple consequences of counting arguments, and will be used in most of the proofs.

\begin{lemma}\label{lemma:Technical}
Let $G$ be a graph, and $\lambda>0$. The three following properties hold :
\begin{enumerate}
    \item \label{lemma:Technical0} $Z_G(\lambda ) \leqslant  \max(\lambda^{\alpha(G)},1)|\mathcal{I}(G)|$ ;
    \item \label{lemma:Technical1} for any $X\subseteq V(G)$ and independent set $I\subseteq X$, 
    $$
    \sum_{\substack{K\in \mathcal{I}(G) \\ K\cap X =I}} \pi_{G,\lambda}(K) \leqslant \lambda^{|I|}\frac{Z_{V(G)\setminus X}(\lambda)}{Z_G(\lambda)};
    $$ 
    \item \label{lemma:Technical2} for any partition $(X,Y)$ of $V(G)$, $Z_G(\lambda) \leqslant Z_{X}(\lambda)Z_{Y}(\lambda)$.
\end{enumerate}
\end{lemma}

\begin{proof}
We first prove point (\ref{lemma:Technical0}). Observe that, given any independent set $I\in \mathcal{I}(G)$, $\lambda^{|I|} \leqslant \lambda^{\alpha(G)}$ if $\lambda \geqslant 1$, and $\lambda^{|I|}\leqslant 1$ if $\lambda \leqslant 1$. Thus, 
$
Z_G(\lambda ) \leqslant \sum_{I\in \mathcal{I}(G)} \max(\lambda^{\alpha(G)},1) 
$

Then, we prove (\ref{lemma:Technical1}). Let $K\in \mathcal{I}(G)$ such that $K\cap X=I$. Notice that $K$ is the union of an independent set of $G[V\setminus X]$, that is $K\setminus X$, with the independent set $I$. Thus, 
$$
\sum_{\substack{K\in \mathcal{I}(G) \\ K\cap X =I}} \lambda^{|K|} \leqslant \sum_{K'\in \mathcal{I}(G[V\setminus X])} \lambda^{|K'|+|I|} = \lambda^{|I|}Z_{V\setminus X}(\lambda)
$$
and the result is obtained by dividing by $Z_G(\lambda)$.

Finally, we prove point (\ref{lemma:Technical2}). Notice that for any independent set $I\in \mathcal{I}(G)$, $I\cap X \in \mathcal{I}(G[X])$ and $I\cap Y\in \mathcal{I}(G[Y])$. Thus, since $(X,Y)$ is a partition of $V(G)$,
\begin{align*}
Z_G(\lambda) &=\sum_{I\in \mathcal{I}(G)}\lambda^{|I|} \\
&\leqslant \sum_{I_X\in \mathcal{I}(G[X])}\sum_{I_Y\in \mathcal{I}(G[Y])} \lambda^{|I_X\cup I_Y|}\\
&= Z_{X}(\lambda) Z_{Y}(\lambda)
\end{align*}
\end{proof}

\subsection{Conductance} The conductance of Markov chain $\mathcal{M}$ is defined as the edge expansion of the underlying transition graph. It is the main tool to prove lower bound on the mixing time of Glauber dynamics. One of the simplest methods to bound the conductance, introduced by Randall~\cite{randall2006slow}, involves partitioning the state space $\mathcal{I}(G)$ into three subsets $(\Omega_S, \Omega_1, \Omega_2)$ such that $\Omega_S$ \emph{separates} $\Omega_1$ and $\Omega_2$. This means that for any pair of states $(I,J) \in \Omega_1 \times \Omega_2$ and any path $\gamma_{I,J} = (I = W_0, W_1, \ldots, W_\ell = J)$ in $\mathcal{G}(G)$, there exists at least one intermediate independent set $W_i$ such that $W_i \in \Omega_S$. For the sake of completeness, we now formally state the following theorem, which encapsulates this argument and will be used in Section~\ref{sec:LowerBounds}.

\begin{theorem}\cite{sinclair1992improved,randall2006slow}\label{thm:LowerBoundConductance}
Let $G$ be a graph and let $(\Omega_S, \Omega_1,\Omega_2)$ be a partition of $\mathcal{I}(G)$ such that $\Omega_S$ separates $\Omega_1$ and $\Omega_2$. For any $\lambda>0$,
$$
\tau_{G,\lambda} \geqslant \frac{\ln 2}{4}\left(\frac{\pi_{G,\lambda}(\Omega_1)}{\pi_{G,\lambda}(\Omega_S)}-2\right)
$$
\end{theorem}

\subsection{Canonical Paths}

The \emph{canonical paths} technique provides a powerful method for upper bounding the mixing time. Let $G$ be a graph, and for each pair of independent sets $I, J \in \mathcal{I}(G)$, fix a specific path $\gamma_{I,J} = (I = W_0, W_1, \ldots, W_\ell = J)$ in $\mathcal{G}(G)$. The path $\gamma_{I,J}$ is called the canonical path from $I$ to $J$. Let $\Gamma := \{\gamma_{I,J} \mid I,J \in \mathcal{I}(G)\}$ be the set of all fixed canonical paths. The \emph{congestion} through a transition $e = (W, W')$ where $|W\Delta W'|\leqslant 1$, is defined as
\begin{center}
\begin{equation}\label{Congestion}
\rho(\Gamma, e) := \frac{1}{\pi_{G,\lambda}(W)P(W, W')} \sum_{\substack{K, L \in \mathcal{I}(G) \\ \gamma_{K,L} \text{ uses } e}} \pi_{G,\lambda}(K) \pi_{G,\lambda}(L) |\gamma_{K,L}|
\end{equation}
\end{center}
where $|\gamma_{K,L}|$ is the length of the path $\gamma_{K,L}$ and $P(W,W')$ is the probability associated to the transition $(W,W')$. The overall congestion of the paths $\Gamma$ is then
$$
\rho(\Gamma) := \max_{e = (W,W') : P(W,W') > 0} \rho(\Gamma, e).
$$
We will use the following :

\begin{theorem}[\cite{sinclair1992improved}]\label{thm:BoundCanonicalPaths}
Let $G$ be a graph and let $\Gamma$ be a set of canonical paths between independent sets of $G$. For any $\lambda>0$,
$$
\tau_{G,\lambda} \leqslant \rho(\Gamma)\ln\left(\frac{4}{\min_{I\in \mathcal{I}(G)}\pi_{G,\lambda}(I)}\right)
$$
\end{theorem}

\section{Subexponential Mixing Time on Geometric Graphs}\label{sec:GeometricGraphs}

This section is dedicated to proving a slightly more general result that will imply Theorem~\ref{thm:BallGraphs}. First, we introduce the notion of clique-based separator trees, which serve as the main tool for our proof. Then, we establish that in any graph belonging to a class that admits clique-based separators, the Glauber dynamics mixes in subexponential time.

\subsection{Clique-based separator and clique-based separator tree}
\paragraph{Clique-based separator} Let $a\in (0,1)$ and $w>0$. Following from \cite{de2018framework}, we say that a triple $(A,B,S)$ is a \emph{$(a,w)$-clique based separation} of $G$ if it verifies the following properties :
\begin{itemize}
    \item $(A,B,S)$ is a partition of $G$ such that $(A\times B) \cap E(G) = \emptyset$ ;
    \item $|A| \leqslant a |V(G)|$ and $|B|\leqslant a |V(G)|$ ;
    \item there exists a set of cliques $C_1,...,C_k$ ($k\geqslant 1$) such that $S= \bigcup_{1\leqslant i \leqslant k} C_i$ and such that 
    $$
    \sum_{1\leqslant i \leqslant k} \log_2(|C_i|+1) \leqslant w
    $$
\end{itemize}
The set $S$ is then called a \emph{$(a,w)$-clique based separator} of $G$. A triple $(A,B,S)$ that satisfies only the first condition is simply called a \emph{separation} of $G$, and $S$ is a \emph{separator} of $G$.

\begin{observation}\label{Obs:CliqueBased}
Suppose that $(A,B,S)$ is a $(a,w)$-clique based separator of a graph $G$. Then $G[S]$ contains at most $2^{w}$ distinct independent sets, each of them of size at most $w$.
\end{observation}

\begin{proof}
Since $C_1,...,C_k$ are cliques, any independent set of $G[S]$ intersects any clique on at most one vertex. It follows that there are at most $\prod_{1\leqslant i \leqslant k} (|C_i|+1)$ independent sets of $G[S]$ which can be upper bounded  by $2^w$ using simple computations. Then, notice that $\alpha(G[S]) \leqslant k$, and that $w=\sum_{1\leqslant i\leqslant k} \log_2(|C_i|+1) \geqslant \sum_{1\leqslant i\leqslant k} 1 = k$.
\end{proof}

Observe that a clique-based separator generalizes the classical notion of a separator, where only the number of vertices in the separator matters. Suppose that $(A, B, S)$ satisfies the first two conditions of the definition of a clique-based separator. Then, it also forms a $(a, |S|)$-clique-based separator by considering only cliques of size $1$. This shows that clique-based separators are particularly interesting when the number of vertices in $S$ is large compared to the total number of vertices in the graph.

\paragraph{Clique-Based Separator tree} Let $a\in (0,1)$ be a constant and let $g$ be a real-valued non-decreasing function.  We say that $\mathcal{T} = (T,\{X_t\}_{t\in V(T)})$  is a \emph{$(a,g)$-clique based separator tree} of a graph $G$ if  $T$ is a rooted binary full tree with the following properties:
\begin{enumerate}
    \item Each node $t \in V(T)$ is associated with a set $X_t \subseteq V(G)$.
    \item The sets $X_t$, for $t \in T$, partition $V(G)$, i.e., 
    $$
    \bigcup_{t \in V(T)} X_t = V(G), \quad \text{and} \quad X_t \cap X_{t'} = \emptyset \text{ for distinct } t, t' \in V(T).
    $$
    \item For each node $t \in V(T)$, let $V_t = \bigcup_{s} X_s$, where $s$ ranges over the descendants of $t$ (including $t$). Note that if $t$ is an internal node with children $u$ and $v$, then $V_t$ is the disjoint union of $X_t$, $V_u$, and $V_v$. If $t$ is a leaf, then $V_t = X_t$.
    
    For each internal node $t \in T$ with children $u$ and $v$, the triplet $(V_u, V_v, X_t)$ forms a $(a,g(|V_t|))$-clique based separation for the subgraph $G[V_t]$.
    \item For each leaf $t \in T$, we have $|V_t| \leqslant 1$.
\end{enumerate}

A $(1,id_{\mathbb{R}})$-clique-based separator tree of $G$ is simply called a \emph{separator tree} of $G$, as the third condition reduces to a separation condition, while the fourth condition becomes trivial.

\begin{observation}\label{Obs:CliqueBasedSeparatorTree}
Let $\mathcal{C}$ be a hereditary class of graphs such that each graph with $n$ vertices admits a $(a,g(n))$-clique based separator for some constant $a$ and function $g$ depending on $\mathcal{C}$. Then, each graph $G\in \mathcal{C}$ admits an $(a,g)$-clique based separator tree.
\end{observation}

\begin{proof}
We prove the result by induction on $n$. If $G$ has at most $1$ vertex, the claim holds trivially by considering a clique-based separator tree with a single node $t$ and setting $X_t = V(G)$.

Now, assume the result holds for all graphs with at most $n-1$ vertices, and let $G$ be a $n$-vertex graph. Since $G \in \mathcal{C}$, it admits an $(a, g(n))$-clique-based separator $(A,B,S)$. By the induction hypothesis, $G[A]$ has a $(a,g)$-clique-based separator tree $\mathcal{T}_A = (T_A, \{X_t\}_{t\in V(T_A)})$, and similarly, $G[B]$ has a clique-based separator tree $\mathcal{T}_B = (T_B, \{X_t\}_{t\in V(T_B)})$.

We construct a separator tree $\mathcal{T}$ for $G$ by introducing a new root node $t_0$ and defining the tree structure as follows: the vertex set of $T$ is $V(T_A) \cup V(T_B) \cup \{t_0\}$, and the edge set is $E(T_A) \cup E(T_B) \cup {(t_0, r_A), (t_0, r_B)}$, where $r_A$ and $r_B$ are the roots of $T_A$ and $T_B$, respectively. Let $X_{t_0}=S$ be the bag associated with the root $t_0$, and then $\mathcal{T} =(T, \{X_t\}_{t\in V(T)})$ is a $(a,g)$-clique-based separator tree of $G$ by definition.
\end{proof}

\subsection{Bounding mixing time}

Using the notion of a separator tree, we first define a set of canonical paths that will exhibit subexponential congestion. This will then lead to an upper bound on the mixing time.

\paragraph*{Definition and properties of canonical paths.} Let $G$ be a graph and $\mathcal{T} = (T, \{X_t\}_{t\in V(T)})$  be a separator tree of $G$. For the moment, $\mathcal{T}$ does not have to be specifically clique-based separator tree. We define $\Gamma_{\mathcal{T},G}$ a set of canonical paths between the independent sets of $G$ as follows. We define by induction on $T$, for any $t\in V(T)$ and $I,J$ independent sets of $G[V_t]$, the canonical path $\gamma_{I,J}$ from $I$ to $J$.
\begin{itemize}
    \item If $t$ is a leaf, then the path $\gamma_{I,J}$ consists in removing all the vertices of $I$ and then add all vertices of $J$ ;
    \item otherwise, $t$ has a left child $u\in V(T)$ and a right child $v\in V(T)$.  $\gamma_{I,J}$ is obtained as follows :
    \begin{itemize}
        \item go from $I$ to $I\setminus X_t$ by deleting all the vertices of $I\cap X_t$;
        \item go from $I\setminus X_t$ to $(J\cap V_u)\cup (I\cap V_v)$ by following the canonical path from $(I\cap V_u)$ to $(J\cap V_u)$ obtained by induction on $u$;
        \item go from $(J\cap V_u)\cup (I\cap V_v)$ to $(J\cap V_u)\cup (J\cap V_v) = J \setminus X_t$ by following the canonical path from $(I\cap V_v)$ to $(J\cap V_v)$ obtained by induction on $v$;
        \item go from $J\setminus X_t$ to $J$ by adding all the vertices of $J\cap X_t$.
    \end{itemize}
\end{itemize}
Notice that for any pair of independent sets $(I,J)$ of $G$, the path $\gamma_{I,J}$ is defined. We set $\Gamma(\mathcal{T},G)$ to be the set of such paths.  For a vertex $v\in V(G)$, let $t_v\in V(T)$ be the unique node of $T$ such that $v\in X_{t_v}$. In addition, let $$A_v =\bigcup_{t\in V(T) \mid \text{$t$ ancestor of $t_v$ in $T$}} X_t
$$ be all the vertices of $G$ which appear in the bag of an ancestor node of $t_v$ in $T$. Note that, in particular, $v\in A_v$.

The key property of this set of canonical paths is that, for any transition $(I, I \cup {v})$ or $(I, I \setminus {v})$, where $I \in \mathcal{I}(G)$ and $v \in V(G)$, the congestion of this transition is bounded by a function of $Z_{A_v}(\lambda)^2$, up to a polynomial factor. The next lemmas formalizes this argument. In the following, we denote by $\pi$ the distribution $\pi_{G,\lambda}$ when there is no ambiguity.

\begin{lemma}\label{lemma:Congestion1}
Let $G$ be a graph, $\mathcal{T}$ a separator tree of $G$, $I$ an independent set of $G$, and $v\in V(G)\setminus N[I]$. Then, for any $\lambda>0$,
$$
\rho(\Gamma_{\mathcal{T},G}, (I,I\cup\{v\}) \leqslant  2n^2\frac{\lambda+1}{\lambda} \frac{1}{\lambda^{|I\cap A_v|}} Z_{A_v}(\lambda)^2 
$$
\end{lemma}
\begin{proof}
By Equation~\ref{Congestion}, we have 
\begin{align*}
    \rho(\Gamma_{\mathcal{T},G}, (I, I\cup\{v\}) &= \frac{1}{\pi(I) P(I,I\cup\{v\})} \sum_{\substack{K,L\in \mathcal{I}(G) \\ (I,I')\in \gamma_{K, L}}} \pi(K)\pi(L) |\gamma_{K, L}| \\
    &\leqslant 2n^2\frac{\lambda+1}{\lambda} \frac{Z_G(\lambda)}{\lambda^{|I|}}\sum_{\substack{K,L\in \mathcal{I}(G) \\ (I,I')\in \gamma_{K, L}}} \pi(K)\pi(L) 
\end{align*}
The first inequality holds since for any $K,L\in \mathcal{I}(G)$, the length of a canonical path is at most $2n$, the probability $P(I,I\cup\{v\})$ is $\frac{1}{n}\frac{\lambda}{\lambda+1}$, and finally $\pi(I) = \frac{\lambda^{|I|}}{Z_G(\lambda)}$. 

Suppose that $(I,I\cup\{v\})\in \gamma_{K,L}$ for some independent sets $K,L \in \mathcal{I}(G)$. Consider the following order on the nodes of $T$ : for $t,t'\in V(T)$, we state $t \prec t'$ if and only if $t$ is strictly before $t'$ in the post-order of $T$. Then, we define two sets :
$$
\begin{array}{ll}
     L_v &= \bigcup_{t \in V(T) \mid t\prec t_v} X_{t} \\
     R_v &= V(G) \setminus (L_v\cup A_v)
\end{array}
$$
Notice that, $(L_v,A_v,R_v)$ forms a partition of the set fo vertices $V(G)$. An illustration of the three sets is given in Figure~\ref{fig:1stCase}.

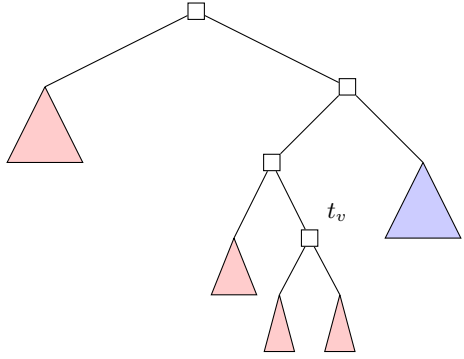
\begin{figure}[h]
\centering
    \begin{tikzpicture}
    \node[draw] (1) at (0,0) {};
    \node[draw] (2) at (2, -1) {};
    \node[draw] (3) at (1, -2) {};
    \node[draw, label=above right:$t_v$] (4) at (1.5, -3) {};
    \draw (1) -- (2)--(3)--(4) ;
    
    \draw (1) -- (-2, -1) ;
    \fill[draw, fill = red!20] (-2, -1)-- (-2.5, -2) -- (-1.5, -2) -- (-2,-1);
    
    \draw (2) -- (3,-2);
    \fill[draw, fill = blue!20] (3,-2) -- (2.5, -3) -- (3.5,-3) -- (3,-2);
    
    \draw (3) -- (0.5,-3);
    \fill[draw, fill = red!20] (0.5, -3) -- (0.2,-3.75) -- (0.8, -3.75) -- (0.5,-3);
    
    \draw (4) -- (1.1,-3.75);
	\fill[draw, fill = red!20] (1.1,-3.75) -- (0.9, -4.5) -- (1.3, -4.5) -- (1.1,-3.75);
    
    \draw (4) -- (1.9,-3.75);
    \fill[draw, fill = red!20] (1.9,-3.75) -- (1.7, -4.5) -- (2.1, -4.5) -- (1.9,-3.75);

    \end{tikzpicture}

    \caption{Illustration of the sets $L_v,A_v$ and $R_v$ in the case of $I'=I\cup \{v\}$. The nodes in red represents $L_v$, the nodes in blue $R_v$ and the remaining represents $A_v$, ie the ancestors of $v$.}
    \label{fig:1stCase}
\end{figure}

Suppose that $(I,I\cup \{v\})$ is on the path from $K$ to $L$. Along the path, all vertices from $K\cap A_v$ have been removed, all vertices from $K\cap L_v$ have been changed to the vertices of $L\cap L_v$, and the vertices of $K \cap R_v$ have not been modified. Thus, $I\cap L_v = L\cap L_v$, $I\cap R_v = K\cap R_v$. Thus, 
$$
\sum_{\substack{K,L\in \mathcal{I}(G) \\ (I,I\cup\{v\})\in \gamma_{K,L}}} \pi(K)\pi(L) \leqslant \left(\sum_{\substack{K\in \mathcal{I}(G) \\ K\cap R_v = I\cap R_v}} \pi(K) \right) 
    \left(\sum_{\substack{L\in \mathcal{I}(G) \\ L\cap L_v = I\cap L_v}} \pi(L) \right)
$$

We bound the first of two sum on the above equation. By applying Lemma~\ref{lemma:Technical}.(\ref{lemma:Technical1}) on $G$, subset $R_v\subseteq V(G)$ and independent set $I\cap R_v$, we have 
$$
\sum_{\substack{K\in \mathcal{I}(G) \\ K\cap R_v = I\cap R_v}} \pi(K) \leqslant \frac{\lambda^{|I\cap R_v|}}{Z_{G}(\lambda)} Z_{V(G) \setminus R_v}
$$
Then, by Lemma~\ref{lemma:Technical}.(\ref{lemma:Technical2}) on graph $G[V(G)\setminus R_v]$ and partition $(L_v,A_v)$ of $V(G)\setminus R_v$, we get 
$$
\sum_{\substack{K\in \mathcal{I}(G) \\ K\cap R_v = I\cap R_v}} \pi(K) \leqslant \frac{\lambda^{|I\cap R_v|}}{Z_G(\lambda)} Z_{L_v}(\lambda) Z_{A_v}(\lambda)
$$
Similarly, we bound the second sum as follows :
$$
\sum_{\substack{L\in \mathcal{I}(G) \\ L\cap L_v = I\cap L_v}} \pi(L) \leqslant \frac{\lambda^{|I\cap L_v|}}{Z_G(\lambda)} Z_{R_v}(\lambda) Z_{A_v}(\lambda)
$$
Combining the inequalities above, we obtain that 
\begin{align*}
    \rho(\Gamma_{\mathcal{T},G}, (I, I\cup \{v\})) 
    &\leqslant 2n^2\frac{\lambda+1}{\lambda}\frac{Z_G(\lambda)}{\lambda^{|I|}} \sum_{\substack{K,L\in \mathcal{I}(G) \\ (I,I\cup\{v\})\in \gamma_{K,L}}} \pi(K)\pi(L) \\
    &\leqslant 2n^2\frac{\lambda+1}{\lambda}  \frac{Z_G(\lambda)}{\lambda^{|I|}} \left(\sum_{\substack{K\in \mathcal{I}(G) \\ K\cap R_v = I\cap R_v}} \pi(K) \right) 
    \left(\sum_{\substack{L\in \mathcal{I}(G) \\ L\cap L_v = I\cap L_v}} \pi(L) \right)\\[1em]
    &\leqslant 2n^2\frac{\lambda+1}{\lambda}\frac{Z_G(\lambda)}{\lambda^{|I|}} \frac{\lambda^{|I\cap L_v|} \lambda^{|I\cap R_v|}}{Z_G(\lambda)^2}
    Z_{L_v}(\lambda) Z_{R_v}(\lambda)Z_{A_v}(\lambda)^2 \\[1em]
    &= 2n^2\frac{\lambda+1}{\lambda} \frac{1}{\lambda^{|I\cap A_v|}}\frac{Z_{L_v}(\lambda) Z_{R_v}(\lambda)}{Z_G(\lambda)} Z_{A_v}(\lambda)^2 
\end{align*}
where the last inequality holds since $\frac{\lambda^{|I\cap L_v|} \lambda^{|I\cap R_v|}}{\lambda^{|I|}} = \frac{1}{\lambda^{|I\cap A_v|}}$.

Notice that $R_v$ and $L_v$ are disconnected in $G$, thus if $K\in \mathcal{I}(G[L_v])$ and $L\in \mathcal{I}(G[R_v])$, then $K\cup L$ is an independent set of $G$. It follows that
$$
\left(\sum_{K \in \mathcal{I}(G[L_v])}  \lambda^{|K|}  \right) 
    \left(\sum_{L \in \mathcal{I}(G[R_v])} \lambda^{|L|} \right) \leqslant \sum_{I\in \mathcal{I}(G)}\lambda^{|I|} = Z_G(\lambda)
$$
and finally, 
$$
\rho(\Gamma, (I, I\cup \{v\}))  \leqslant 2n^2\frac{\lambda+1}{\lambda} \frac{1}{\lambda^{|I\cap A_v|}} Z_{A_v}(\lambda)^2 
$$
\end{proof}

\begin{lemma}\label{lemma:Congestion2}
Let $G$ be a graph, $\mathcal{T}$ a separator tree of $G$, $I$ an independent set of $G$, and $v\in I$. Then, for any $\lambda>0$,
$$
\rho(\Gamma_{\mathcal{T},G}, (I,I\setminus\{v\}) \leqslant
 2n^2\frac{\lambda+1}{\lambda^{|I\cap A_v|}}Z_{A_v}(\lambda)^2
$$
\end{lemma}

\begin{proof}
Similarly to the proof of Lemma~\ref{lemma:Congestion1}, by Equation~\ref{Congestion} we obtain 
\begin{align*}
    \rho(\Gamma_{\mathcal{T},G}, (I, I\setminus\{v\}) &= \frac{1}{\pi(I) P(I,I\setminus\{v\})} \sum_{\substack{K,L\in \mathcal{I}(G) \\ (I,I')\in \gamma_{K, L}}} \pi(K)\pi(L) |\gamma_{K, L}| \\
    &\leqslant 2n\frac{Z_G(\lambda)}{\lambda^{|I|}}n(\lambda+1) \sum_{\substack{K,L\in \mathcal{I}(G) \\ (I,I')\in \gamma_{K, L}}} \pi(K)\pi(L) \\
\end{align*}

Consider the following order on the nodes of $T$ : for $t,t'\in V(T)$, we state $t \prec t'$ if and only if $t$ is strictly before $t'$ in the pre-order of $T$. Then, we define two sets :
$$
\begin{array}{ll}
     R_v &= \bigcup_{t \in V(T) \mid t_v\prec t} X_{t} \\
     L_v &= V(G) \setminus (A_v\cup R_v)
\end{array}
$$
Note that $(L_v,A_v,R_v)$ is a separation of $V(G)$. An illustration of the three sets is given in Figure~\ref{fig:2ndCase}.

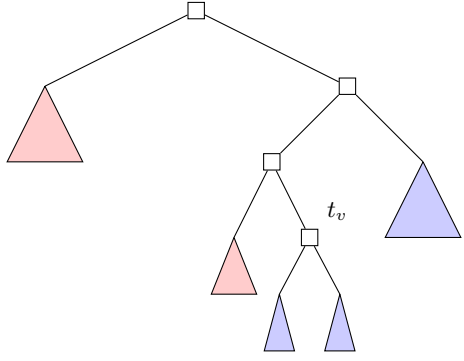
\begin{figure}[h]
\centering
    \begin{tikzpicture}
    \node[draw] (1) at (0,0) {};
    \node[draw] (2) at (2, -1) {};
    \node[draw] (3) at (1, -2) {};
    \node[draw, label=above right:$t_v$] (4) at (1.5, -3) {};
    \draw (1) -- (2)--(3)--(4) ;
    
    \draw (1) -- (-2, -1) ;
    \fill[draw, fill = red!20] (-2, -1)-- (-2.5, -2) -- (-1.5, -2) -- (-2,-1);
    
    \draw (2) -- (3,-2);
    \fill[draw, fill = blue!20] (3,-2) -- (2.5, -3) -- (3.5,-3) -- (3,-2);
    
    \draw (3) -- (0.5,-3);
    \fill[draw, fill = red!20] (0.5, -3) -- (0.2,-3.75) -- (0.8, -3.75) -- (0.5,-3);
    
    \draw (4) -- (1.1,-3.75);
	\fill[draw, fill = blue!20] (1.1,-3.75) -- (0.9, -4.5) -- (1.3, -4.5) -- (1.1,-3.75);
    
    \draw (4) -- (1.9,-3.75);
    \fill[draw, fill = blue!20] (1.9,-3.75) -- (1.7, -4.5) -- (2.1, -4.5) -- (1.9,-3.75);

    \end{tikzpicture}

    \caption{Illustration of the sets $L_v,A_v$ and $R_v$ in the case of $I'=I\setminus \{v\}$. The nodes in red represents $L_v$, the nodes in blue $R_v$ and the remaining represents $A_v$, ie the ancestors of $v$.}
    \label{fig:2ndCase}
\end{figure}

Suppose that $(I,I\setminus \{v\})$ is on the path from $K$ to $L$. Along the path, all vertices from $K\cap A_v$ have been removed, all vertices from $K\cap L_v$ have been changed to the vertices of $L\cap L_v$, and the vertices of $K \cap R_v$ have not been modified. Thus, $I\cap L_v = L\cap L_v$ and $I\cap R_v = K\cap R_v$. 

Following the same exact step of computation as in Lemma~\ref{lemma:Congestion1}, we obtain that :
$$
\sum_{\substack{K,L\in \mathcal{I}(G) \\ (I,I')\in \gamma_{K, L}}} \pi(K)\pi(L) \leqslant
\frac{\lambda^{|I\cap L_v|}\lambda^{|I\cap R_v|}}{Z_G(\lambda)}\underbrace{\frac{Z_{L_v}(\lambda)Z_{R_v}(\lambda)}{Z_G(\lambda)}}_{\leqslant 1}Z_{A_v}(\lambda)^2
$$
and finally, 
$$
\rho(\Gamma_{\mathcal{T},G}, (I, I\setminus\{v\}) \leqslant 2n^2 \frac{\lambda+1}{\lambda^{|I\cap A_v|}} Z_{A_v}(\lambda)^2
$$
\end{proof}

\begin{lemma}\label{lemma:Congestion3}
Let $G$ be a graph, $\mathcal{T}=(T,\{X_t\}_{t\in V(T)})$ be a separator tree of $G$. Then, for any $\lambda>0$, 
$$
\rho(\Gamma_{\mathcal{T},G}) \leqslant 4n^2\tilde{\lambda}^{2\alpha+1}\max_{v\in V(G)}|\mathcal{I}(G[A_v])|^2
$$
where  $\tilde{\lambda} = e^{|\ln \lambda|}$ and $\alpha = \max_{v\in V(G)} \alpha(G[A_v])$.
\end{lemma}

\begin{proof}[Proof of Lemma~\ref{lemma:Congestion3}]
Let $(I,I')$ be a transition of the Markov chain. By definition of $\Gamma := \Gamma_{\mathcal{T},G}$, no canonical path uses the transition $(I,I)$, since at each step, either a vertex is added to or removed from the current independent set. It follows that $\rho(\Gamma, (I,I))=0$. Thus, we only need to consider the cases where either $I' = I \cup \{v\}$ for some $v \in V(G) \setminus N[I]$ or $I' = I \setminus \{v\}$ for some $v \in I$. A direct application of Lemma~\ref{lemma:Congestion1} and Lemma~\ref{lemma:Congestion2} gives  
$$
\rho(\Gamma,(I,I')) \leqslant 2n^2 \max\left(\lambda+1,\frac{\lambda+1}{\lambda}\right) \frac{1}{\lambda^{|I\cap A_v|}}Z_{A_v}(\lambda)^2
$$
First, suppose that $\lambda \geqslant 1$. In this case, the equality $\tilde{\lambda} = \lambda$ holds, and by Lemma~\ref{lemma:Technical}.(\ref{lemma:Technical0}),  
$$
Z_{A_v}(\lambda) \leqslant \tilde{\lambda}^{\alpha(G[A_v])} |\mathcal{I}(G[A_v])|
$$
Since $\frac{1}{\lambda^{|I\cap A_v|}} \leqslant 1$ and $\lambda+1 \leqslant 2\lambda$, we obtain:  
$$
\rho(\Gamma,(I,I')) \leqslant 4n^2 \tilde{\lambda}^{2\alpha(G[A_v])+1} |\mathcal{I}(G[A_v])|^2
$$
Next, suppose that $\lambda \in (0,1)$. In this case, the equality $\tilde{\lambda} = \lambda^{-1}$ holds. Since $\lambda^{\alpha(G[A_v])} \leqslant 1$, Lemma~\ref{lemma:Technical}.(\ref{lemma:Technical0}) implies $Z_{A_v}(\lambda) \leqslant |\mathcal{I}(G[A_v])|$. Furthermore, we have $\max\left(\lambda+1, \frac{\lambda+1}{\lambda} \right) \leqslant 2\tilde{\lambda}$, and $\frac{1}{\lambda^{|I\cap A_v|}} \leqslant \tilde{\lambda}^{\alpha(G[A_v])}$. It follows that  
$$
\rho(\Gamma,(I,I')) \leqslant 4n^2 \tilde{\lambda}^{\alpha(G[A_v])+1} |\mathcal{I}(G[A_v])|^2
$$
The conclusion is straightforward.

\end{proof}

\paragraph*{Main result.} Lemma~\ref{lemma:Congestion3} gives an upper bound on the congestion of $\Gamma_{\mathcal{T}, G}$ depending on $\max_{v\in V(G)} \alpha(G[A_v])$ and $\max_{v\in V(G)} |\mathcal{I}(G[A_v])|$. However, this result does not rely on any specific properties beyond $\mathcal{T}$ being a separator tree. Theorem~\ref{thm:SubexpMixingTime} leverages the fact that $\mathcal{T}$ is a clique-based separator tree of $G$ to derive a subexponential bound on the congestion, and consequently, on the mixing time.

\begin{theorem}\label{thm:SubexpMixingTime}
Let $\mathcal{C}$ be a hereditary class of graphs, such that any $n$-vertex graph $G\in \mathcal{C}$ has a $(a,cn^{1-b}\ln^\delta(n))$-clique based separator for some constants $a$, $b$, $c$ and $\delta$ depending on~$\mathcal{C}$. Let $G\in \mathcal{C}$ be a $n$-vertex graph. The mixing time of the Glauber dynamics on $G$ with fugacity $\lambda> 0$ satisfies :
$$
\tau_{G,\lambda} = (2\tilde{\lambda})^{O(n^{1-b}\ln^\delta(n))}.
$$
\end{theorem}

\begin{proof} Let $g : \mathbb{R} \rightarrow \mathbb{R}$ be the function defined by $g(x) = cx^{1-b}\ln^\delta(x)$, and let $\mathcal{T} = (T, \{X_t\}_{t\in V(T)})$ be a $(a,g)$-clique based separator tree of $G$.

We show that the congestion is bounded on the set of canonical paths $\Gamma :=\Gamma_{\mathcal{T},G}$. By Lemma~\ref{lemma:Congestion3}, we have 
$$
\rho(\Gamma) \leqslant 4n^2 \tilde{\lambda}^{1+2\max_{v\in V(G)} \alpha(G[A_v])} \max_{v\in V(G)}|\mathcal{I}(G[A_v])|^2
$$
Consider $v\in V(G)$, and we bound both $\alpha(G[A_v])$ and $|\mathcal{I}(G[A_v])|$.
We describe how to construct an independent set of $G[A_v]$. Let $t_0...t_{k-1}$ be the ancestors of $t_v$ in $T$, where $t_0$ is the root and $t_{k-1}=t_v$, and notice that $X_{t_0},...,X_{t_{k-1}}$ forms a partition of $A_v$. Since any independent set $J \in \mathcal{I}(G[A_v])$ is the union of $k$ disjoint independent sets of $G[X_{t_0}],...,G[X_{t_{k-1}}]$ respectively, we have both inequalities
$$
\begin{array}{rl}
\alpha(G[A_v]) &\leqslant \sum_{i=0}^{k-1} \alpha(G[X_{t_i}]) \\[1em]
|\mathcal{I}(G[A_v])| &\leqslant \prod_{i=0}^{k-1} |\mathcal{I}(G[X_{t_i}])|
\end{array}
$$
Using Observation~\ref{Obs:CliqueBased}, there at most $2^{g(n)}$ independent sets in $G[X_{t_0}]$, and each of them has size at most $g(n)$.  Similarly, since $X_{t_1}$ is a clique based separator in a graph with less than $an$ vertices, $G[X_{t_1}]$ contains at most $2^{g(an)}$ independent sets of size at most $g(an)$. Repeating the same reasoning, we obtain that 
$$
\alpha(G[A_v]) \leqslant \sum_{i=0}^{k-1} g(a^in)
$$
and also
$$
|\mathcal{I}(G[A_v])|  \leqslant \prod_{i=0}^{k-1} 2^{g(a^in)}
= 2^{\sum_{i=0}^{k-1} g(a^in)}
$$
The sum can be bounded above :
\begin{align*}
    \sum_{i=0}^{k-1} g(a^in) &= \sum_{i=0}^{k-1} c(a^in)^{1-b}\ln^\delta(a^in) \\
    &= cn^{1-b}\sum_{i=0}^{k-1} (a^{1-b})^i\ln^\delta(a^in) \\
    &\leqslant cn^{1-b}\ln^\delta(n) \sum_{i=0}^{k-1} (a^{1-b})^i\\
    &\leqslant \frac{c}{1-a^{1-b}}n^{1-b}\ln^\delta(n) 
\end{align*}
where the last inequality holds since $a^{1-b} <1$ and $\sum_{i=0}^{+\infty}(a^{1-b})^i = \frac{1}{1-a^{1-b}}$. Altogether, the congestion verifies: 
$$
\rho(\Gamma) \leqslant 4n^2 (2\tilde{\lambda})^{2\frac{c}{1-a^{1-b}} n^{1-b}\ln^\delta(n)+1} \\
$$
Finally, using Theorem~\ref{thm:BoundCanonicalPaths} and $\min_{I\in \mathcal{I}(G)} \pi_{G,\lambda}(I) \geqslant (2\tilde{\lambda})^{-n}$,
\begin{align*}
    \tau_{G,\lambda} &\leqslant 4n^2(2\tilde{\lambda})^{2\frac{c}{1-a^{1-b}} n^{1-b}\ln^\delta(n)+1}\ln\left( 4(2\tilde{\lambda})^n\right) \\[1em]
    &\leqslant (2\tilde{\lambda})^{\gamma_{a,b,c}n^{1-b}\ln^\delta(n)}
\end{align*}
where $\gamma_{a,b,c}$ is a large enough constant depending only on $a$, $b$ and $c$. 
\end{proof}

We now discuss implications of Theorem~\ref{thm:SubexpMixingTime} together with previous works on clique-based separators from \cite{de2018framework,de2023clique}. A geometric object in $\mathbb{R}^d$ is said to be \emph{convex} if, for any two points inside the object, the entire line segment connecting them also lies within the object. A set of geometric objects is \emph{fat} if there exists a constant $c>0$ such that each object $O$ contains a ball of radius $r$ and is contained in a ball of radius at most $c \cdot r$. In particular, similarly-sized fat objects are those where the ratio between the largest and smallest enclosing balls is bounded. A \emph{map graph} is a graph obtained as follows: consider a planar graph $H$ with vertices representing regions in a planar map. The map graph is then constructed by connecting two vertices if their corresponding regions share a boundary. These graphs generalize planar graphs and may have unbounded clique sizes. A family of regions in the plane forms a collection of \emph{pseudo-disks} if the boundaries of any two regions intersect at most twice. That is, for any two regions, their intersection consists of at most two connected components. Pseudo-disk intersection graphs generalize disk graphs.

\begin{corollary}
Let $G$ be a graph with $n$ vertices. The mixing time of the Glauber dynamics on $G$ with fugacity $\lambda > 0$ satisfies:
\begin{itemize}
    \item If $G$ is the intersection graph of $d$-dimensional convex fat objects, or similarly-sized fat objects, then 
        $
        \tau_{G,\lambda} = (2\tilde{\lambda})^{O(n^{1-1/d})}.
        $
    \item If $G$ is the intersection graph of pseudo-disks, then
        $
        \tau_I{G,\lambda} = (2\tilde{\lambda})^{O(n^{\frac{2}{3}}\ln(n))}.
        $
    \item If $G$ is a map graph, then
        $
        \tau_{G,\lambda} = (2\tilde{\lambda})^{O(\sqrt{n})}.
        $
\end{itemize}
\end{corollary}

The proof of Theorem~\ref{thm:BallGraphs} follows directly, as ball graphs in $\mathbb{R}^d$ are intersection graphs of convex objects. 

\section{Parameterized Mixing Time}\label{sec:Parameterized}

This section is dedicated to the proof of Theorem~\ref{thm:PathTree}. All the upper bounds on the mixing time follow the same outline as in the previous section: we construct a set of canonical paths, which in this case depend on path decompositions rather than separator trees. We then show that the congestion of these paths is bounded by the desired function of the parameter. Finally, we extend these results to the tree independence number by establishing structural connections between path and tree decompositions.

\paragraph*{Path independence number and pathwidth.}
Let $\mathcal{P}=\{X_t\}_{1\leqslant t \leqslant p}$ be a path decomposition of a graph $G$, where $p$ is the number of bags in the path decomposition. Up to adding a bag, we suppose that $X_p$ is empty. Given two independent sets $I,J\in \mathcal{I}(G)$, we define the canonical path from $I$ to $J$ as follows. The step $1$ consists in, starting from $I$,  removing all the vertices of $I\cap X_1$ in arbitrary order. Then, for any $2\leqslant i \leqslant p$ successively, step $i$ consists in removing the vertices from $I\cap X_{i}$ in arbitrary order, and adding all the vertices from $J\cap X_{i-1}\setminus X_i$ in arbitrary order as well.  We call $\gamma_{I,J}$ the set of transitions obtained. It is straightforward to prove that $\gamma_{I,J}$ is indeed a path from $I$ to $J$, since each vertex $v\in I$ is removed at step $t_0$ where $t_0 = \min \{i \mid v\in X_i\}$, and any vertex $v\in J$ is added at step $t_1$ where $t_1= \min \{i\mid v\in X_{i-1} \wedge v\notin X_i\}$. We call $\Gamma_{\mathcal{P},G}$ the set of canonical path $\gamma_{I,J}$ obtained for all pair of independent sets $I,J\in \mathcal{I}(G)$.

The next lemma is the equivalent of Lemma~\ref{lemma:Congestion1} for path decomposition instead of separator tree.
\begin{lemma}\label{lemma:CongestionPath1}
Let $G$ be a graph, $\mathcal{P} = \{X_t\}_{1\leqslant t\leqslant p}$ a path decomposition of $G$ with $X_p=\emptyset$, $I\in \mathcal{I}(G)$ and $v\notin N[I]$. In addition, let $t = \min \{i \mid v\in X_{i-1} \wedge v\notin X_i\}$. Then, for any $\lambda> 0$,
$$
\rho(\Gamma_{\mathcal{P},G}, (I, I\cup \{v\})) \leqslant 2n^2\frac{\lambda+1}{\lambda}\frac{1}{\lambda^{|I\cap X_t|}} Z_{X_t}(\lambda)^2
$$
\end{lemma}

\begin{proof}
By Equation~\ref{Congestion} and similarly to in the proof of Lemma~\ref{lemma:Congestion1}, we have 
$$
\rho(\Gamma_{\mathcal{P},G}, (I, I\cup \{v\}) 
     \leqslant 2n^2\frac{\lambda+1}{\lambda} \frac{Z_G(\lambda)}{\lambda^{|I|}}\sum_{\substack{K,L\in \mathcal{I}(G) \\ (I,I')\in \gamma_{K, L}}} \pi(K)\pi(L) 
$$
Suppose that $(I,I\cup\{v\})\in \gamma_{K,L}$ for some independent sets $K,L \in \mathcal{I}(G)$. We define two sets 
$$
\begin{array}{ll}
     L_t &= \bigcup_{1\leqslant i \leqslant t-1} X_{i} \setminus X_t\\
     R_t &= V(G) \setminus (L_v\cup X_t)
\end{array}
$$

Suppose that $(I,I\cup \{v\})$ is on the path $\gamma_{K,L}$. This transition has to be used in the step $t$ of the path, when the vertices of $X_{t-1}\setminus X_t$ are added to the current independent set. In particular, all the steps from $1$ to $t-1$ have been done, and thus all the vertices from $K\cap \bigcup_{1\leqslant j \leqslant t-1}X_j \setminus X_t$ have been replaced by the vertices of $L\cap \bigcup_{1\leqslant j \leqslant t-1}X_j \setminus X_t$. It follows that $I\cap L_t = L$. Similarly, all the steps from $t+1$ to $p$ have yet to be executed, and thus the vertices of $K\cap \bigcup_{t+1\leqslant j \leqslant p} X_j \setminus X_t$ have not been modified. It follows that $I\cap R_t = K\cap R_t$. We obtain that 
$$
\sum_{\substack{K,L\in \mathcal{I}(G) \\ (I,I\cup\{v\})\in \gamma_{K,L}}} \pi(K)\pi(L) \leqslant \left(\sum_{\substack{K\in \mathcal{I}(G) \\ K\cap R_t = I\cap R_t}} \pi(K) \right) 
    \left(\sum_{\substack{L\in \mathcal{I}(G) \\ L\cap L_t = I\cap L_t}} \pi(L) \right)
$$
As in the proof of Lemma~\ref{lemma:Congestion1}, a successive application of Lemmas~\ref{lemma:Technical}.(\ref{lemma:Technical1}) and~\ref{lemma:Technical}.(\ref{lemma:Technical2}) gives :
$$
\sum_{\substack{K\in \mathcal{I}(G) \\ K\cap R_t = I\cap R_t}} \pi(K) \leqslant \frac{\lambda^{|I\cap R_t|}}{Z_G(\lambda)}Z_{L_t}(\lambda)Z_{X_t}(\lambda) 
$$
$$
\sum_{\substack{K\in \mathcal{I}(G) \\ K\cap L_t = I\cap L_t}} \pi(K) \leqslant \frac{\lambda^{|I\cap L_t|}}{Z_G(\lambda)}Z_{R_t}(\lambda)Z_{X_t}(\lambda)
$$
By combining the inequalities, it follows
$$
\rho(\Gamma_{\mathcal{P},G}, (I, I\cup \{v\})) \leqslant 2n^2\frac{\lambda+1}{\lambda} \frac{\lambda^{|I\cap L_t|+|I\cap R_t|}}{\lambda^{|I|}} \frac{Z_{L_t}(\lambda) Z_{R_t}(\lambda)Z_{X_t}(\lambda)^2 }{ Z_G(\lambda)}
$$
Finally, by Lemma~\ref{lemma:Technical}.(\ref{lemma:Technical2}), we obtain 
$$
\rho(\Gamma_{\mathcal{P},G}, (I, I\cup \{v\})) \leqslant 2n^2\frac{\lambda+1}{\lambda}\frac{1}{\lambda^{|I\cap X_t|}} Z_{X_t}(\lambda)^2
$$
\end{proof}

Similarly, the next lemma is the equivalent of Lemma~\ref{lemma:Congestion2} for path decomposition instead of separator tree.
\begin{lemma}\label{lemma:CongestionPath2}
Let $G$ be a graph, $\mathcal{P} = \{X_t\}_{1\leqslant t\leqslant p}$ a path decomposition of $G$ with $X_p=\emptyset$, $I\in \mathcal{I}(G)$ and $v\in I$. In addition, let $t = \min \{i \mid v\in X_i\}$. Then, for any $\lambda> 0$,
$$
\rho(\Gamma_{\mathcal{P},G}, (I, I\setminus \{v\})) \leqslant 2n^2(\lambda+1)\frac{1}{\lambda^{|I\cap X_t|}} Z_{X_t}(\lambda)^2
$$
\end{lemma}

\begin{proof}
By Equation~\ref{Congestion} and similarly to in the proof of Lemma~\ref{lemma:Congestion2}, we have 
$$
\rho(\Gamma_{\mathcal{P},G}, (I, I\setminus\{v\}) 
    \leqslant 2n^2 \frac{Z_G(\lambda)}{\lambda^{|I|}}(\lambda+1) \sum_{\substack{K,L\in \mathcal{I}(G) \\ (I,I')\in \gamma_{K, L}}} \pi(K)\pi(L)
$$
Suppose that $(I,I\setminus\{v\})\in \gamma_{K,L}$ for some independent sets $K,L \in \mathcal{I}(G)$. We define two sets 
$$
\begin{array}{ll}
     L_t &= \bigcup_{1\leqslant i \leqslant t-1} X_{i} \setminus X_t\\
     R_t &= V(G) \setminus (L_v\cup X_t)
\end{array}
$$

Suppose that $(I,I\setminus \{v\})$ is on the path $\gamma_{K,L}$. This transition is used in the step $t$ of the path, that is when the vertices of $K\cap X_t$ are removed. Similarly to the proof of Lemma~\ref{lemma:CongestionPath1}, we have  $I\cap L_t = L$ and $I\cap R_t = K\cap R_t$. Again, we obtain that, 
$$
\sum_{\substack{K,L\in \mathcal{I}(G) \\ (I,I\setminus\{v\})\in \gamma_{K,L}}} \pi(K)\pi(L) \leqslant \left(\sum_{\substack{K\in \mathcal{I}(G) \\ K\cap R_t = I\cap R_t}} \pi(K) \right) 
    \left(\sum_{\substack{L\in \mathcal{I}(G) \\ L\cap L_t = I\cap L_t}} \pi(L) \right)
$$
Following the exact same proof of Lemma~\ref{lemma:CongestionPath1}, the result holds.
\end{proof}

The next Lemma is the equivalent of Lemma~\ref{lemma:Congestion3} for path decomposition instead of separator tree.

\begin{lemma}\label{lemma:CongestionPath3}
Let $G$ be a graph, $\mathcal{P} = \{X_t\}_{1\leqslant t\leqslant p}$ a path decomposition of $G$ with $X_p=\emptyset$. Then, for any any $\lambda>0$,
$$
\rho(\Gamma_{\mathcal{P},G}) \leqslant 4n^2\tilde{\lambda}^{2\alpha+1}\max_{1\leqslant t\leqslant p}|\mathcal{I}(G[X_t])|^2
$$
where $\alpha = \max_{1\leqslant t\leqslant p} \alpha(G[X_t])$.
\end{lemma}

\begin{proof}[Proof of Lemma~\ref{lemma:CongestionPath3}]
Let $(I,I')$ be a transition of the Markov chain. By definition of $\Gamma:= \Gamma_{\mathcal{P},G}$, no canonical path use the transition $(I,I)$, since at each step either a vertex is added or is deleted to the current independent set. Thus, we only consider the cases $I'=I\cup \{v\}$ for some $v\in V(G)\setminus N[I]$ or $I'= I\setminus \{v\}$ for some vertex $v\in I$. In both cases, by Lemma~\ref{lemma:CongestionPath1} and Lemma~\ref{lemma:CongestionPath2}, there exists a $t \in \{1,\cdots, p\}$ such that 
$$
\rho(\Gamma,(I,I')) \leqslant 2n^2\max\left(\lambda+1, \frac{\lambda}{\lambda+1} \right)\frac{1}{\lambda^{|I\cap X_t|}}Z_{X_t}(\lambda)^2 
$$
First, suppose that $\lambda \geqslant 1$. In this case, the equality $\tilde{\lambda} = \lambda$ holds, and by Lemma~\ref{lemma:Technical}.(\ref{lemma:Technical0}),  
$$
Z_{X_t}(\lambda) \leqslant \tilde{\lambda}^{\alpha(G[X_t])} |\mathcal{I}(G[X_t])|
$$
Since $\frac{1}{\lambda^{|I\cap X_t|}} \leqslant 1$ and $\lambda+1 \leqslant 2\lambda$, we obtain:  
$$
\rho(\Gamma,(I,I')) \leqslant 4n^2 \tilde{\lambda}^{2\alpha(G[X_t])+1} |\mathcal{I}(G[X_t])|^2
$$
Next, suppose that $\lambda \in (0,1)$. In this case, the equality $\tilde{\lambda} = \lambda^{-1}$ holds. Since $\lambda^{\alpha(G[X_t])} \leqslant 1$, Lemma~\ref{lemma:Technical}.(\ref{lemma:Technical0}) implies $Z_{A_v}(\lambda) \leqslant |\mathcal{I}(G[X_t])|$. Furthermore, we have $\max\left(\lambda+1, \frac{\lambda+1}{\lambda} \right) \leqslant 2\tilde{\lambda}$, and $\frac{1}{\lambda^{|I\cap X_t|}} \leqslant \tilde{\lambda}^{\alpha(G[X_t])}$. It follows that  
$$
\rho(\Gamma,(I,I')) \leqslant 4n^2 \tilde{\lambda}^{\alpha(G[X_t])+1} |\mathcal{I}(G[X_t])|^2
$$
The conclusion is straightforward.
\end{proof}

We are now equipped to prove the first case of Theorem~\ref{thm:PathTree}, together with a FPT mixing time for graphs of bounded pathwidth.

\begin{theorem}\label{thm:Path}
Let $G$ be a graph with $n$ vertices. The mixing time of the Glauber dynamics with fugacity $\lambda> 0$ satisfies 
\begin{itemize}
     \item $\tau_{G,\lambda} = (2\tilde{\lambda})^{O(\pw(G))} \cdot n^{O(1)}$ ;
    \item $\tau_{G,\lambda} = (b\tilde{\lambda})^{O(\pa(G))} \cdot n^{O(1)}$ where $b$ is the minimum width (plus one) of a path decomposition of path independence number $\pa(G)$.
\end{itemize}
\end{theorem}

\begin{proof}
Let $\mathcal{P}=\{X_t\}_{1\leqslant t \leqslant p}$ be a path decomposition of $G$, where $p$ is the number of bags in the path decomposition and $X_p = \emptyset$. We bound the congestion of the set of canonical paths $\Gamma :=\Gamma_{\mathcal{P},G}$. The reasoning follows exactly as in the proof of Theorem~\ref{thm:SubexpMixingTime}. By Lemma~\ref{lemma:CongestionPath3}, we have
$$
\rho(\Gamma) \leqslant 4n^2 \tilde{\lambda}^{1+2\max_{1\leqslant t \leqslant p} \alpha(G[X_t])}\max_{1\leqslant t\leqslant p}|\mathcal{I}(G[X_t])|^2
$$
 Then, we bound both $\alpha(G[X_t])$ and $|\mathcal{I}(G[X_t])|$ for all $1\leqslant t \leqslant p$, depending on whether the path decomposition is a witness for $\pw(G)$ or $\pa(G)$.

Firstly, suppose that, for any $1\leqslant t\leqslant p$, $|X_t|\leqslant \pw(G)+1$. It follows that $G[X_t]$ contains at most $2^{\pw(G)+1}$ distinct independent sets, each of size at most $\pw(G)+1$. Thus, 
$$
\rho(\Gamma) \leqslant 4n^2 \tilde{\lambda}^{2\pw(G)+3}2^{2\pw(G)+2} \leqslant 4n^2(2\tilde{\lambda})^{2\pw(G)+3}
$$
Using Theorem~\ref{thm:BoundCanonicalPaths} and $\min_{I\in \mathcal{I}(G)} \pi_{G,\lambda}(I) \geqslant (2\tilde{\lambda})^{-n}$,
$$
\tau_{G,\lambda} \leqslant 4n^2(2\tilde{\lambda})^{2\pw(G)+3} \ln\left( 4(2\tilde{\lambda})^{-n}\right)
$$
and finally $\tau_{G,\lambda} = (2\tilde{\lambda})^{2\pw(G)+4} \cdot n^{O(1)}$.

Secondly, suppose that, for any $1\leqslant t\leqslant p$, $\alpha(G[X_t])\leqslant \pa(G)$ and $|X_t| \leqslant b$. It follows that $G[X_t]$ contains at most $b^{\pa(G)}$ distinct independent sets, each of size at most $\pa(G)$. Thus, 
$$
\rho(\Gamma) \leqslant 4n^2 \tilde{\lambda}^{2\pa(G)+1}b^{2\pa(G)} \leqslant 4n^2(b\tilde{\lambda})^{2\pa(G)+1}
$$
Using Theorem~\ref{thm:BoundCanonicalPaths}, we obtain that $\tau_{G,\lambda} = (b\tilde{\lambda})^{2\pa(G)+2} \cdot n^{O(1)}$.
\end{proof}

\paragraph{Tree independence number.}
To establish bounds on the mixing time parameterized by the tree independence number, there are two possible approaches. One option is to use a separator-tree-based method, as in Section~\ref{sec:GeometricGraphs}. The other is to relate the width of path decompositions to the width of tree decompositions. A result of this type is already known for sparse parameters: for any $n$-vertex graph, it is known that $\pw(G) = O(\ln(n))\tw(G)$~\cite{korach1993tree}. We provide a strengthening of this result for the dense case.

\begin{lemma}\label{lemma:structural_bounds}
Let $G$ be a graph with $n$ vertices, and let $\mathcal{T}=(T,\{X_t\}_{t\in V(T)})$ be a tree decomposition of $G$ such that $|V(T)|\leqslant 4n$ and for any $t\in V(T)$, $|X_t|\leqslant b$ and $\alpha(G[X_t])\leqslant k$. Then, there exists a path decomposition $\mathcal{P}=\{Y_i\}_{1\leqslant i \leqslant p}$ such that for each $1\leqslant i \leqslant p$ :
\begin{itemize}
    \item $|Y_i|\leqslant b(\log_2(n)+3)$ ;
    \item $\alpha(G[Y_i])\leqslant k(\log_2 n + 3)$ ;
    \item  $|\mathcal{I}(G[Y_i])|\leqslant b^{(\log_2 n + 3)k}$.
\end{itemize}
\end{lemma}

\begin{proof}
We begin by proving the following well-known bound on the pathwidth of trees.

\begin{claim}
For any tree $T$ with $n$ vertices, $\pw(T) \leqslant \log_2 n$.
\end{claim}

\begin{proof}
We proceed by induction on $n$. The base case is trivial for trees with at most two vertices. Assume the statement holds for all trees with at most $n-1$ vertices. 

Let $T$ be a tree with $n$ vertices. Define a vertex $r \in V(T)$ such that the maximum size of a connected component of $T - r$ is minimized. Let $T_1, \dots, T_k$ be these connected components. By the choice of $r$, each $T_i$ has at most $n/2$ vertices; otherwise, the only vertex adjacent to $r$ in $T_i$ would contradict the minimality condition. By the induction hypothesis, for each $1\leqslant i \leqslant k$, we have $\pw(T_i) \leqslant \log_2(n/2)$. 

Concatenating the $k$ path decompositions of $T_1, \dots, T_k$ and adding $r$ to each bag results in a path decomposition of $T$ of width at most $\log_2(n/2) + 1 = \log_2 n$, completing the induction.
\end{proof}

Now, we use this claim to establish both statements of the lemma. By the previous claim, $T$ admits a path decomposition $\mathcal{P}_{T}=(U_i)_{1\leqslant i \leqslant p}$ such that for any $1\leqslant i \leqslant p$, $|U_i| \leqslant \log_2(4n)+1 = \log_2 n + 3$. For each $1\leqslant i \leqslant p$, define $Y_i = \bigcup_{t\in U_i}X_t$. Then, $\mathcal{P}=(Y_i)_{1\leqslant i \leqslant p}$ is a path decomposition of $G$. Let $1\leqslant i \leqslant p$. 
\begin{itemize}
    \item Since for any $t\in U_i$, $|X_t|\leqslant b$, it follows that $|Y_i| \leqslant b(\log_2(n)+3)$.
    \item Any independent set of $G[Y_i]$ is the disjoint union of independent sets $I_t$, where each $I_t \subseteq X_t$ for some $t\in U_i$. Since $\alpha(G[X_t]) \leqslant k$ for $t\in U_i$, it follows that $\alpha(G[Y_i])\leqslant k(\log_2 n + 3)$. 
    \item Similarly, since $G[X_t]$ has at most $t^k$ independent sets for $t\in U_i$, it follows that $|\mathcal{I}(G[Y_i])|\leqslant b^{k(\log_2 n + 3)}$. 
\end{itemize}
\end{proof}

Using these inequalities, we derive new bounds on the mixing time of the Glauber dynamics for independent sets in terms of treewidth and tree independence number. The first case of the following theorem corresponds to the result of Eppstein and Frishberg \cite{eppstein2023rapid}. The second result, which is the exactly the second case of Theorem~\ref{thm:PathTree}, generalizes the result of Bez\'akov\'a and Sun \cite{bezakova2020mixing}, as chordal graphs are exactly graph of tree independence number $1$.

\begin{theorem}\label{thm:generalized_mixing_time}
Let $G$ be a graph with $n$ vertices. The mixing time of the Glauber dynamics on $G$ with fugacity $\lambda > 0$ satisfies:
\begin{itemize}
    \item $\tau_{G,\lambda} = n^{O(\tw(G))\ln(\tilde{\lambda})} $ ;
    \item $\tau_{G,\lambda} = n^{O(\ta(G))\cdot \ln(b\tilde{\lambda})} $, where $b$ is the minimum width (plus one) of a tree decomposition of tree independence number $\ta(G)$.
\end{itemize}
\end{theorem}

\begin{proof}
The first case is a direct consequence of Theorem~\ref{thm:Path} and Lemma~\ref{lemma:structural_bounds}. We focus on the second case. Let $\mathcal{P}=\{X_i\}_{1\leqslant i \leqslant p}$ be a path decomposition of a graph $G$, where $p$ is the number of bags in the path decomposition and $X_p = \emptyset$. In addition, from Lemma~\ref{lemma:structural_bounds}, suppose that for each $1\leqslant i \leqslant p$, we have $\alpha(G[X_i])\leqslant (\log_2(n)+3)\ta(G)$ and $|\mathcal{I}(G[X_i])|\leqslant b^{(\log_2(n)+3)\ta(G)}$.

We show that the congestion is bounded on the set of canonical paths $\Gamma :=\Gamma_{\mathcal{P},G}$. The reasoning follows exactly as in the proof of Theorem~\ref{thm:Path}. By Lemma~\ref{lemma:CongestionPath3}, we have
\begin{align*}
\rho(\Gamma) &\leqslant 4n^2 \tilde{\lambda}^{1+2\max_{1\leqslant t \leqslant p} \alpha(G[X_t])}\max_{1\leqslant t\leqslant p}|\mathcal{I}(G[X_t])|^2  \\
&\leqslant 4n^2\tilde{\lambda}^{1+2(\log_2(n)+3)\ta(G)} b^{2(\log_2(n)+3)\ta(G)}
\end{align*}
Using Theorem~\ref{thm:BoundCanonicalPaths}, we obtain that $\tau_{G,\lambda} = n^{O(\ta(G))\cdot \ln(b\tilde{\lambda})}$.
\end{proof}

\section{Lower Bound on the Mixing Time of Glauber Dynamics}\label{sec:LowerBounds}

This section is dedicated to reviewing and establishing lower bounds on the mixing time of the Glauber dynamics, based on upper bounds on the conductance of the Markov chain. We demonstrate that our parameterized mixing time bounds for dense graphs are tight by constructing, in each case, a family of graphs that matches the upper bounds stated in Theorems~\ref{thm:PathTree}. 

\subsection{Geometric intersection graphs}

Most of the research on Glauber dynamics has focused on the distinction between fast mixing and torpid mixing, where the latter corresponds to an exponential mixing time. Lattices provide a prototypical example of graphs exhibiting torpid mixing. In their seminal work, the authors of~\cite{borgs1999torpid} studied Glauber dynamics on the $d$-dimensional torus and established the following result:

\begin{theorem}[\cite{borgs1999torpid}]\label{thm:LowerBoundTorus}
Let $T_{L,d}$ be the $d$-dimensional torus with side length $L$. For $d\geqslant 2$ and $\lambda$ sufficiently large, the mixing time of the glauber dynamcis on $T_{L,d}$ verifies :
$$
\tau_{T_{L,d},\lambda} \geqslant e^{cL^{d-1}/\ln^2(L)}
$$
where $c$ is a constant depending only on $d$.
\end{theorem}

This theorem already demonstrates some tightness in the bound of Theorem~\ref{thm:SubexpMixingTime}. Indeed, let $\mathcal{H}(T_{L,d})$ be the set of induced subgraphs of $T_{L,d}$, which forms a hereditary class by definition. It is easy to observe that any $n$-vertex graph $G \in \mathcal{H}(T_{L,d})$ has a balanced separator consisting of at most $\sqrt{n}$ vertices, and thus admits a $(1/2, \sqrt{n})$-clique-based separator. Consequently, in the statement of Theorem~\ref{thm:SubexpMixingTime}, it is not possible to improve the exponent of $n$.

However, it is important to note that the 2-dimensional torus (resp. the $d$-dimensional torus) is not a disk graph (resp. a $d$-dimensional ball graph), and thus Theorem~\ref{thm:LowerBoundTorus} does not directly provide a lower bound for Theorem~\ref{thm:BallGraphs}. Nevertheless, for sufficiently large $L$, the torus $T_{L,d}$ is a $(d+1)$-dimensional ball graph. As a result, the exponent in Theorem~\ref{thm:BallGraphs} cannot be improved beyond $(d-2)/d$.

In the two-dimensional case, Randall~\cite{randall2006slow} established a stronger lower bound on the mixing time of Glauber dynamics on the 2-dimensional grid. Using the direct implication of Theorem~\ref{thm:LowerBoundConductance}, it follows :

\begin{theorem}[\cite{randall2006slow}]
Let $\Lambda_{L,2}$ be the $2$-dimensional grid with side length $L$. For  $\lambda\geqslant 58.612$, the mixing time of the glauber dynamics on $\Lambda_{L,2}$ verifies :
$$
\tau_{\Lambda_{L,2},\lambda} \geqslant e^{cL}
$$
where $c$ is a constant.
\end{theorem}

The generalization to the $d$-dimensional grid remains to be proven but appears to hold. Notably, in the case of the Ising model, which shares many similarities with the hard-core model, Thomas~\cite{thomas1989bound} demonstrated that the mixing time is indeed exponential in $n^{1-1/d}$ on the $d$-dimensional grid with $n$ vertices.

\subsection{Pathwidth and path independence number.} 

Using a conductance argument, the simplest graph on which Glauber dynamics mixes in exponential time is the complete bipartite graph $K_{t,t}$. Indeed, to transition from an independent set contained entirely in one part to an independent set contained in the other, the dynamics must pass through the empty independent set, which has significantly lower energy compared to any other state. Observe that this graph has pathwidth exactly $t$. We show that this graph establishes the tightness of the FPT mixing time bound in terms of pathwidth.

\begin{theorem}\label{thm:LowerBoundPw}
For any $c>0$, there exists a $n$-vertex graph $G$  such that for any $\lambda\geqslant 1$,  
$$
\tau_{G,\lambda} \geqslant  \left(\frac{\lambda+1}{2}\right)^{\pw(G)} \cdot n^{c} .
$$
\end{theorem}

\begin{lemma}\label{lemma:PiEdgeless}
Let $G$ be an edgeless graph with $n$ vertices. Then, for any $\lambda>0$, $Z_G(\lambda)=(\lambda+1)^n$.
\end{lemma}

\begin{proof}
For any $0\leqslant k \leqslant n$, $G$ has $\binom{n}{k}$ independent sets of size exactly $k$. It follows that 
$$
Z_G(\lambda) = \sum_{k=0}^n \binom{n}{k}\lambda^k = (\lambda+1)^n.
$$
\end{proof}

\begin{proof}[Proof of Theorem~\ref{thm:LowerBoundPw}]
Let  $\Omega_t = \mathcal{I}(K_{t,t})$ as the set of independent sets of $K_{t,t}$. It is straightforward to verify that $K_{t,t}$ has pathwidth exactly $t$. Let $(A_t,B_t)$ be the unique bipartition of $V(G_t)$ such that both $A_t$ and $B_t$ are maximum independent sets of $G_t$. Define $\Omega_A$ (resp. $\Omega_B$) as the set of independent sets of $G_t$ contained entirely within $A_t$ (resp. $B_t$). 

Notice that any path from $\Omega_A$ to $\Omega_B$ in $\mathcal{G}(K_{t,t})$ must pass through the empty set. Otherwise, such a path would contain an independent set which intersect both $A_t$ and $B_t$, contradicting the definition of an independent set. By Theorem~\ref{thm:LowerBoundConductance}, and using the fact that $\pi(\Omega_A) = \pi(\Omega_B)$ by symmetry, we obtain :
$$
\tau_{G,\lambda} \geqslant \frac{\ln(2)}{4}\left( \frac{\pi(\Omega_A)}{\pi(\{\emptyset\})}-2\right) 
$$  
Since $\pi(\emptyset) = \frac{1}{Z_G(\lambda)}$ and $\pi(\Omega_A) = \frac{(\lambda+1)^t}{Z_G(\lambda)}$ by Lemma~\ref{lemma:PiEdgeless}, it follows that
$$
\tau_{G,\lambda} \geqslant \frac{\ln(2)}{4}(\lambda+1)^t
$$  
Then, observe that there exists some constant $t_c\geqslant 1$ such that for any $t\geqslant t_c$, the following inequality holds:
$$
\frac{\frac{\ln(2)}{4}(\lambda+1)^{t}}{\left(\frac{\lambda+1}{2}\right)^t\cdot (2t)^c } = \frac{\ln(2)}{4}2^{t-c\log_2(2t)}>1
$$
which finally leads to 
$$
\tau_{G,\lambda} \geqslant \left(\frac{\lambda+1}{2}\right)^{t} \cdot (2t)^c 
$$  
when $t$ is large enough.
\end{proof}

By replacing each vertex in $K_{t,t}$ with a clique, we establish the tightness of the mixing time in terms of the path independence number.

\begin{theorem}\label{thm:LowerBoundPa}
For any $t\geqslant 1$, there exists infinitely many graphs $G$ such that $\pa(G)=t$ and for any $\lambda \geqslant 1$,
$$
\tau_{G,\lambda} \geqslant c_t\left(n\lambda\right)^{t}
$$
where $n=|V(G)|$ and $c_t$ is a constant which only depends on $t$.
\end{theorem}

\begin{lemma}\label{lemma:P3-free}
Let $G$ be a graph which is the disjoint union of $t$ cliques of size $n$. Then, for any $\lambda>0$, $Z_G(\lambda)=(n\lambda+1)^t$.
\end{lemma}

\begin{proof} 
For any $0\leqslant k \leqslant t$, $G$ has exactly $n^k \binom{t}{k}$ independent sets of size exactly $k$. Indeed, any independent set must intersect each clique in at most one vertex, meaning that choosing an independent set corresponds to selecting $k$ cliques and picking exactly one vertex inside each chosen clique. It follows that:
$$
Z_G(\lambda) = \sum_{k=0}^t \binom{t}{k}n^k \lambda^k = (n\lambda+1)^t
$$
\end{proof}

\begin{lemma}[Folklore]\label{lemma:Neighboroodtree}
For any graph $G$ and tree decomposition $\mathcal{T}=(T,\{X_t\}_{t\in V(T)})$ of $G$, there exist $v\in V(G)$ and $t\in V(T)$ such that $N[v] \subseteq X_t$.
\end{lemma}

\begin{proof}[Proof of Theorem~\ref{thm:LowerBoundPa}]
Consider the graph $G_{p,t}$ constructed as follows:

\begin{itemize} \item Let $G_1$ be the disjoint union of $t$ cliques, each of size $p$, and let $G_2$ be an identical copy of $G_1$. \item Add all possible edges between $G_1$ and $G_2$. \end{itemize}

We first prove that $G_{p,t}$ has a path-independence number exactly $t$. Let $C_1, \dots, C_t$ be the $t$ cliques of size $p$ in $G_2$. Observe that the collection of sets $\{V(G_1) \cup C_i\}_{1\leqslant i \leqslant t}$ forms a valid path decomposition of $G_{p,t}$, and for each $1\leqslant i \leqslant t$, $\alpha(G[V(G_1) \cup C_i]) = t$. Furthermore, for any path decomposition $\{X_i\}_{1\leqslant i \leqslant k}$ of $G_{p,t}$, by Lemma~\ref{lemma:Neighboroodtree}, there exists a vertex $v\in V(G)$ and an index $1\leqslant i \leqslant k$ such that $N[v] \subseteq X_i$. Since for any vertex $v\in V(G)$, its closed neighborhood induces a graph where $\alpha(N[v])=t$ (as the neighborhood of each vertex is obtained by adding all the edges between a clique of size $p-1$ and a disjoint union of $t$ cliques), it follows that $\pa(G_{p,t}) = t$.

Now, define $\Omega_{p,t} = \mathcal{I}(G_{p,t})$ and let $\Omega_1$ (resp. $\Omega_2$) denote the independent sets contained entirely in $V(G_1)$ (resp. $V(G_2)$).

Notice that any transition from $\Omega_1$ to $\Omega_2$ must pass through the empty set. By Theorem~\ref{thm:LowerBoundConductance}, and using the fact that $\pi(\Omega_1) = \pi(\Omega_2)$ by symmetry, we obtain:
$$
\tau_{G,\lambda}\geqslant \frac{\ln(2)}{4}\left( \frac{\pi(\Omega_1)}{\pi(\{\emptyset \})}-2\right) 
$$  
Since $\pi(\emptyset) = \frac{1}{Z_G(\lambda)}$ and by Lemma~\ref{lemma:P3-free}, $\pi(\Omega_1) = \frac{(n\lambda+1)^t}{Z_G(\lambda)}$, it follows that
$$
\tau_{G,\lambda} \geqslant \frac{\ln(2)}{8}(p\lambda+1)^t \geqslant \frac{\ln(2)}{8}(p\lambda)^t 
$$  
After substituting $p$ with $\frac{n}{2t}$, where $n$ is the number of vertices in $G_{p,t}$, we derive:
$$
\tau_{G,\lambda} \geqslant \frac{\ln(2)}{4(2t)^t}(n\lambda)^t
$$  
This completes the proof.
\end{proof}

\subsection{Tree-independence number} Having proven Theorems~\ref{thm:LowerBoundPw} and~\ref{thm:LowerBoundPa}, it is natural to ask whether these results extend to establish the optimality of the mixing time for treewidth and tree-independence number. For treewidth, an equivalent question was already raised by Eppstein and Frishberg~\cite{eppstein2023rapid}, who asked whether an FPT mixing time in terms of treewidth is possible. For the tree-independence number, Dyer, Greenhill, and Müller posed the weaker question of whether the quasi-polynomial mixing time bound for bounded bipartite treewidth is tight. We answer this in the affirmative: by showing that the quasi-polynomial mixing time is optimal for graphs with bounded tree-independence number, we resolve their question. This follows from Proposition~\ref{prop:BoundBipartite}, since graphs of bounded tree-independence number also have bounded bipartite treewidth.

The proof strategy for Theorems~\ref{thm:LowerBoundPw} and~\ref{thm:LowerBoundPa} can be summarized as follows: identify a graph that has exactly two distinct "high-energy states" (such as two distinct maximum independent sets) and demonstrate that any transformation from one state to the other must pass through a low-energy state. A natural first step is to find such a graph that also has bounded treewidth. This idea has been previously explored in the context of reconfiguration~\cite{de2018independent} and has also been applied in the analysis of the Metropolis Markov chain~\cite{nardi2016hitting}. \newline
 
We construct a family of graphs $G_k$ for $k\geqslant 0$. Let $T_k$ be the complete binary tree of depth $k$. The graph $G_k$ is obtained by replacing each node $t\in V(T_k)$ with a set $X_t = \{u_t,v_t\}$ of two adjacent vertices. For each edge $tt'\in E(T_k)$, we add a matching of two edges $u_tu_{t'}$ and $v_tv_{t'}$. The graph $G_3$ is illustrated in Figure~\ref{fig:G3}. Building on the work from \cite{de2018independent}, the graph $G_k$ can be seen as the simplest graph of treewidth $2$ with unbounded TAR reconfiguration threshold. We next prove some basic results about those graphs.

\begin{figure}[t]
    \centering
    \begin{tikzpicture}[scale=0.8]
    \foreach \i in {0,...,7}{
        \node[draw, circle, fill] () at (\i,0) {};
        \node[draw, circle,fill] () at (\i,-1) {};
        \draw (\i,0) --(\i,-1);
    }
    \foreach \i in {0,...,3}{
        \node[draw, circle,fill] () at (2*\i+0.5,1) {};
        \node[draw, circle,fill] () at (2*\i+0.5,-2) {};
        \draw (2*\i+0.5,1)  -- (2*\i+0.5,-2);
        \foreach \j in {2*\i,2*\i+1}{
                \draw (2*\i+0.5,1) --(\j,0); 
                \draw (2*\i+0.5,-2) --(\j,-1); 
        }
    }
    \foreach \i in {0,1}{
        \node[draw, circle,fill] () at (4*\i+1.5,2) {};
        \node[draw, circle,fill] () at (4*\i+1.5,-3) {};
        \draw (4*\i+1.5,2) -- (4*\i+1.5,-3);
        \foreach \j in {4*\i+0.5,4*\i+2.5}{
                \draw (4*\i+1.5,2) --(\j,1); 
                \draw (4*\i+1.5,-3) --(\j,-2); 
        }
    }
    
    \node[draw, circle,fill] () at (3.5,3) {};
    \node[draw, circle,fill] () at (3.5,-4) {};
    \draw (3.5,3)--(3.5,-4);
    \draw (1.5,2) -- (3.5,3) -- (5.5,2);
    \draw (1.5,-3) -- (3.5,-4) -- (5.5,-3);

    \end{tikzpicture}
    \begin{tikzpicture}[scale=0.5]
    \foreach \i in {0,...,7}{
        \node[draw, circle, fill, scale=0.6] () at (2*\i,0) {};
        \node[draw, circle, fill, scale=0.6] () at (2*\i+1,0) {};
        \draw(2*\i+0.5,0) ellipse (0.8 and 0.5);
        \draw (2*\i,0) --(2*\i+1,0);
    }
    
    \foreach \i in {0,...,3}{
        \node[draw, circle, fill, scale=0.6] () at (4*\i+1,2) {};
        \node[draw, circle, fill, scale=0.6] () at (4*\i+1+1,2) {};
        \draw(4*\i+1.5,2) ellipse (0.8 and 0.5);
        \draw (4*\i+1,2) --(4*\i+2,2);
        \draw (4*\i,0) -- (4*\i+1,2) -- (4*\i+2,0);
        \draw (4*\i+1,0) -- (4*\i+2,2) -- (4*\i+3,0);
    }
    
    \foreach \i in {0,1}{
        \node[draw, circle, fill, scale=0.6] () at (8*\i+1+2,4) {};
        \node[draw, circle, fill, scale=0.6] () at (8*\i+1+3,4) {};
        \draw(8*\i+3.5,4) ellipse (0.8 and 0.5);
        \draw (8*\i+3,4) --(8*\i+4,4);
        \draw (8*\i+1,2) -- (8*\i+3,4) -- (8*\i+5,2);
        \draw (8*\i+2,2) -- (8*\i+4,4) -- (8*\i+6,2);
    }
    
    \node[draw, circle, fill, scale=0.6] () at (7,6) {};
    \node[draw, circle, fill, scale=0.6] () at (8,6) {};
    \draw(7.5,6) ellipse (0.8 and 0.5);
    \draw (7,6) --(8,6);
    \draw (3,4) -- (7,6) -- (11,4);
    \draw (4,4) -- (8,6) -- (12,4);

    \node[draw = none] () at (0,-2) {};
    \end{tikzpicture}
    \caption{Two representations of the graph $G_3$. In the right representation, two vertices are in the same bag if and only if they belong to the same set $X_t$ for some $t \in V(T_3)$.}
    \label{fig:G3}
\end{figure}
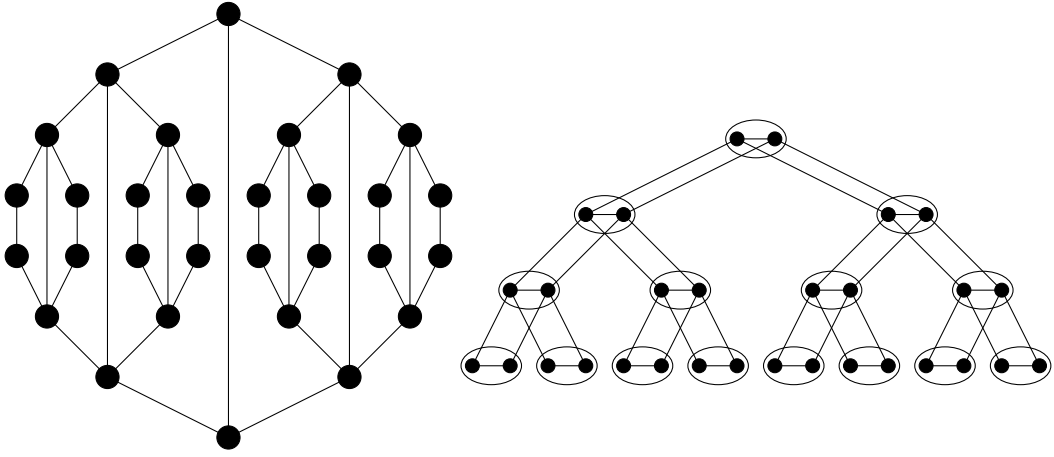

\begin{lemma}\label{lemma:PropertiesGk}
Let $k\geqslant 0$. The following properties hold.
\begin{enumerate}
    \item\label{sublemma:TwoMIS} The graph $G_k$ has $2^{k+2}-2$ vertices and exactly two distinct maximum independent set (MIS) $I_k$ and $J_k$ of size $\alpha(G_k) = 2^{k+1}-1$, and they form a partition of $V(G_k)$.
    \item\label{sublemma:SmallIS} There is a constant $c_1$ such that for every path $(W_i)_{0\leqslant i \leqslant p}$ from $I_k$ to $J_k$ in $\mathcal{G}(G_k)$, there exists some $0\leqslant i \leqslant p$ such that $|W_i| \leqslant \alpha(G_k) - \lfloor c_1 \cdot k \rfloor$.
    \item\label{sublemma:CountIS} For any $d\geqslant 0$, $G_k$ contains at most $2^{2d+1}\alpha(G_k)^d$ independent sets of size exactly $\alpha(G_k) - d$.
\end{enumerate}
\end{lemma}

\begin{proof}
We start by proving point (\ref{sublemma:TwoMIS}). We prove the result by induction on $k$. The base case $k=0$ is straightforward: $G_0$ consists of two adjacent vertices, so it has $|V(G_0)| = 2$ and exactly one maximum independent set of size $\alpha(G_0) = 1$.

Now, suppose the statement holds for $G_{k-1}$. By construction, $G_k$ is obtained from $G_{k-1}$ by adding a new root node $r$ to the binary tree, along with its corresponding set $X_t = {u_t, v_t}$.  By the induction hypothesis, both copies of $G_{k-1}$ in $G_k$, that we call $G_1$ and $G_2$, in $G_k$ contains exactly two distinct MIS, each of size $2^{k}-1$. Let $r_1 \in V(T_k)$ (resp. $r_2$) be the child of $r$ in $T_k$ which correspond to the root bag of $G_1$ (resp. $G_2$). For $i\in \{1,2\}$, let $I_{i,u}$ be the MIS of $G_1$ which contain $u_{r_i}$, and let $I_{i,v}$ be the MIS of $G_2$ that contains $v_{r_i}$. In $G_k$, we construct two MIS as follows:
\begin{itemize}
\item The first MIS, denoted $I_k$, is exactly $\{u_r\}\cup I_{1,v} \cup I_{2,v}$.
\item The second MIS, denoted $J_k$, is exactly $\{v_r\}\cup I_{1,u} \cup I_{2,u}$.
\end{itemize}
Since every vertex belongs to exactly one of these two independent sets, they partition $V(G_k)$. Additionally, the size of each MIS in $G_k$ satisfies:
$$
\alpha(G_k) = 1+2\alpha(G_{k-1}) = 1+2(2^{k}-1) = 2^{k+1}-1
$$
Thus, the lemma holds for all $k \geqslant 0$.\newline

Then, we prove point (\ref{sublemma:SmallIS}), i.e. that any path from $I_k$ to $J_k$ in the reconfiguration graph $\mathcal{G}(G_k)$ pass through a small independent set, using the isoperimetric value of a complete tree.

\begin{claim}[\cite{bharadwaj2009bounds}]\label{lemma:BinaryTree}
There exists a universal constant $c_1 > 0$ such that  
$$
\max_{1\leqslant i\leqslant |V(T_k)|} \min_{\substack{S\subseteq V(T_k) \\ |S|=i}} |N_{T_k}(S)| \geqslant c_1 \cdot k.
$$
\end{claim}
Let $\text{iso}(k)$ be the index at which the maximum is attained.

For each $0\leqslant i \leqslant p$, define $J_i = W_i\cap J_k$. Observe that for all $1\leqslant i <p$, we have $|J_{i+1}|\leqslant |J_i|+1$, with $|J_0|=0$ and $|J_t|=n$. Let $i_0$ be the smallest index such that $|J_{i_0}| = \text{iso}(k)$, and define $S = \{t\in V(T_k) \mid X_t\cap W_{i_0} \neq \emptyset\}$. 

By Claim~\ref{lemma:BinaryTree}, we have $|N_{T_k}(S)|\geqslant c_1\cdot k$. Let $t' \in N_{T_k}(S)$, and note that $X_{t'}\cap W_{i_0} =\emptyset$. Indeed, if $X_{t'}\cap W_{i_0} \cap J_k \neq \emptyset$, then $t'\in S$, contradicting our assumption. Similarly, if $X_{t'}\cap W_{i_0} \cap I_k \neq \emptyset$, then $W_i$ would contain two adjacent vertices, which is impossible. Thus, $W_{i_0}$ intersects at most $|V(T_k)| - \lfloor c_1\cdot k \rfloor$ sets $X_t$, implying that $|W_{i_0}| \leqslant 2^{k+1}-1 - \lfloor c_1 \cdot k \rfloor$. 

Since for any $0\leqslant i <p$, the symmetric difference between $W_i$ and $W_{i+1}$ has size at most~$1$, there must exists some $i$ such that $|W_i| = \alpha(G_k) - \lfloor c_1 \cdot k \rfloor$.\newline

Finally, we prove point (\ref{sublemma:CountIS}). Note that this result was not necessary to prove that TAR-reconfiguration threshold is unbounded in \cite{de2018independent}, but in our case, obtaining such a bound is necessary to apply a conductance argument.

Let $D \subseteq V(T_k)$ be a subset of $d$ nodes of $T_k$. Since $T_k$ is a tree, removing all vertices in $D$ results in a forest. Let $R_D$ be the set of nodes that serve as the roots of the connected components in the forest $T_k - D$. More formally, a node $t \in V(T_k) \setminus D$ belongs to $R_D$ if and only if either $t$ is the child of a node $t' \in D$ or $t$ is the root of $T_k$.

We now prove that there exists a bijection between the independent sets of $G_k$ of size $\alpha(G_k) - d$ and the set $\mathcal{D} = \{(D,\{0,1\}^{R_D}\}_{D\subseteq V(T_k), |D|=d}$.

Given a subset $D \subseteq V(T_k)$ with $|D| = d$ and a function $c_D : R_D \to {0,1}$, we define an independent set $I = f(D, c_D)$ of $G_k$ as follows: 
\begin{itemize}
    \item $I$ intersects each bag $X_t$ for all $t \in V(T_k) \setminus D$.
    \item For each $t \in R_D$, if $c_D(t) = 0$, then $u_t \in I$; otherwise, $v_t \in I$.
\end{itemize}

Now, we prove that there exists a unique independent set $I$ of size $\alpha(G_k) - d$ that satisfies the two conditions above. We will use the following claim. 

\begin{claim}\label{claim:UniqueIS}
Let $T$ be a subtree of $T_k$ with root $r\in V(T)$. Then $G_k[\bigcup_{t\in V(T)} X_t]$ has exactly two independent sets of size $|V(T)|$, one containing $u_r$ and the other containing $v_r$.
\end{claim}

\begin{proof}
We prove the result by induction on $|V(T)|$.  If $T$ consists of a single node $r$, then $X_r = \{u_r, v_r\}$, and the only two independent sets of size $1$ are $\{u_r\}$ and $\{v_r\}$, which satisfies the claim.  Suppose the result holds for all subtrees of $T_k$ with at most $m-1$ nodes. Let $T$ be a subtree of $T_k$ with $m \geqslant 2$ nodes, rooted at $r\in V(T)$. Suppose that $r$ has two children, $r_1$ and $r_2$, both belonging to $V(T)$. Let $T_1$ and $T_2$ be the subtrees of $T$ rooted at $r_1$ and $r_2$, respectively.  

Since both $T_1$ and $T_2$ have at most $m-1$ nodes, we can apply the induction hypothesis to them. This means that:  
\begin{itemize}
\item $G_k[\bigcup_{t\in V(T_1)} X_t]$ has exactly two independent sets of size $|V(T_1)|$, one containing $u_{r_1}$, called $I_{1,u}$ and the other containing $v_{r_1}$, called $I_{1,v}$.  
\item $G_k[\bigcup_{t\in V(T_2)} X_t]$ has exactly two independent sets of size $|V(T_2)|$, one containing $u_{r_2}$, called $I_{2,u}$ and the other containing $v_{r_2}$, called $I_{2,v}$.  
\end{itemize}
To construct an independent set $I$ of $G_k[\bigcup_{t\in V(T)} X_t]$ of size $|V(T)|$, we must include exactly one vertex from each bag $X_t$ for every $t \in V(T)$. 

Thus, we have exactly two choices for constructing $I$:  
\begin{itemize}
\item If $u_r \in I$, then we must take $v_{r_1} \in I$ and $v_{r_2} \in I$ since $u_r$ is adjacent to both vertices $u_{r_1}$ and $u_{r_2}$. By induction, this forces us to take the unique independent sets $I_{1,v}$ and $I_{2,v}$ corresponding to $T_1$ and $T_2$ that contain $v_{r_1}$ and $v_{r_2}$, respectively. Hence,  
  $$
  I = \{u_r\} \cup I_{1,v} \cup I_{2,v}.
  $$
\item Similarly, if $v_r \in I$, then we must take $u_{r_1} \in I$ and $u_{r_2} \in I$, leading to the unique independent sets $I_{1,u}$ and $I_{2,u}$ that contain $u_{r_1}$ and $u_{r_2}$, respectively. Hence,  
  $$
  I = \{v_r\} \cup I_{1,u} \cup I_{2,u}.
  $$  
\end{itemize}
If $r$ has only one child $r_1$ belonging to $V(T)$, then a similar reasoning applies: the only two independent sets of size $|V(T)|$ must contain either $u_r$ or $v_r$, and they uniquely extend to $I_{1,v}$ or $I_{1,u}$, respectively, by the induction hypothesis.

Since these are the only two ways to construct an independent set of size $|V(T)|$, the claim follows.  
\end{proof}

Consider a tree $T$ in the forest $T_k - D$, and let $r \in R_D$ be its root. Consider the independent set $I_T=I\cap~\bigcup_{t\in V(T)} X_t$. Since $I$ intersects all bags $X_t$ for $t\in V(T)$, it follows that $I_T$ is an independent set of $G_k[\bigcup_{t\in V(T)} X_t]$ of size $|V(T)|$. By Claim~\ref{claim:UniqueIS}, the graph $G_k[\bigcup_{t\in V(T)} X_t]$ has exactly two independent sets of size $|V(T)|$: one containing $u_r$, denoted by $I_{T,u}$, and the other containing $v_r$, denoted by $I_{T,v}$.

If $c_D(r) = 0$, then $I \cap X_r = {u_r}$, which implies $I_T = I_{T,u}$. Otherwise, if $c_D(r) = 1$, then $I \cap X_r = {v_r}$, leading to $I_T = I_{T,v}$. Thus, for each tree $T$ in the forest $T_k - D$, there is a unique choice for the independent set $I$ restricted to the bags corresponding to $T$. Applying this argument to all trees in the forest ensures that $I$ is uniquely determined.

It follows that the function $f$ is indeed a bijection. It remains to count the number of elements in $\mathcal{D}$. There are exactly $\binom{|V(T_k)|}{d}$ ways to choose a subset $D \subseteq V(T_k)$ of size $d$. Given such a subset $D$, we observe that the set $R_D$ has at most $2d+1$ elements. Indeed, there are at most $2d$ nodes of $T_k$ whose parent belongs to $D$, and additionally, the root of $T_k$ is included in $R_D$ whenever it does not belong to $D$. This results in at most $2^{2d+1}$ possible choices for functions in $\{0,1\}^{R_D}$.

Thus, the total number of pairs $(D, c_D)$ is at most $\binom{|V(T_k)|}{d}2^{2d+1}$ in total. Finally, we can bound $\binom{|V(T_k)|}{d}$ by $|V(T_k)|^d$ and use the fact that $\alpha(G_k) = |V(T_k)|$ to conclude the proof.
\end{proof}

\begin{proof}[Proof of Theorem~\ref{thm:LowerBoundTreeAlpha}]
Consider the graph $H_{k,t}$ obtained by replacing each vertex $v\in V(G_k)$ with a set of vertices $Y_v$ inducing a disjoint union of $t$ cliques of size $2^{2k}$. For each edge $uv\in E(G_k)$, we add all possible edges between $Y_u$ and $Y_v$. For any subset $V'\subseteq V(H_{k,t})$, define $p(V') = \{u\in V(G_k) \mid Y_u\cap S\neq \emptyset\}$, and observe that if $V'$ is an independent set of $H_{k,t}$, $p(V')$ is an independent set of $G_k$. Let $I_{k,t}$ be a maximum independent set of $H_{k,t}$ such that $p(I_{k,t}) =I_k$. Notice that $|I_{k,t}|=t\alpha(G_k)$. We partition $\mathcal{I}(H_{k,t})$ as follows:
\begin{itemize}
\item $\Omega_S$: independent sets $S$ such that $p(S)$ is an independent set of $G_k$ of size exactly $s =  \alpha(G_k) - \lfloor c_1 k \rfloor$, where $c_1$ is the constant of Lemma~\ref{lemma:PropertiesGk}.(\ref{sublemma:SmallIS}).
\item $\Omega_I$: independent sets $I$ for which there exists a path of independent sets in $\mathcal{G}(H_{k,t})$ from $I_{k,t}$ to $I$ avoiding $\Omega_S$.
\item $\Omega_J$: the remaining independent sets.
\end{itemize}
By definition, $\Omega_S$ separates $\Omega_I$ and $\Omega_J$ in the reconfiguration graph $\mathcal{G}(H_{k,t})$. Indeed, suppose there exists a path from an independent set $I \in \Omega_I$ to an independent set $J \in \Omega_J$ that avoids $\Omega_S$. Then, by the definition of $\Omega_I$, there must also be a path from $I_{k,t}$ to $J$ that avoids $\Omega_S$. This implies that $J \in \Omega_I$, contradicting the assumption that $J \in \Omega_J$.

We denote by $\pi$ the distribution $\pi_{H_{k,t},\lambda}$. To bound the conductance, we first establish an upper bound on $\pi(\Omega_S)$ and a lower bound on $\min(\pi(\Omega_I), \pi(\Omega_J))$.

\begin{claim}\label{claim:LowerBoundOmegaS}
$\pi(\Omega_S) \leqslant 2^{\lfloor c_1k \rfloor (k+3)+1}((2^{2k}\lambda +1)^{t}-1)^s$.
\end{claim}

\begin{proof}
By Lemma~\ref{lemma:PropertiesGk}.(\ref{sublemma:CountIS}), there are at most $2^{2\lfloor c_1 k \rfloor+1} \alpha(G_k)^{\lfloor c_1 k \rfloor}$ independent sets of size $\alpha(G_k) - \lfloor c_1 k \rfloor$.  Let $J \in \mathcal{I}(G_k)$ be an independent set of size $s$. Observe that any independent set $I\in \mathcal{I}(H_{k,t})$ such that $p(I)=J$ is uniquely defined by the $|J|$ pairwise disjoint independent sets $\{I_v\}_{v\in J}$ such that $I_v\in \mathcal{I}(H_{k,t}[Y_v])\setminus \{\emptyset\}$. Thus, it follows that  
$$
\sum_{\substack{I \in \mathcal{I}(H_{k,t}) \\ p(I) = J }} \lambda^{|I|} = \prod_{v\in J} \left( \sum_{I\in \mathcal{I}(H_{k,t}[Y_v]) \setminus \{\emptyset\}} \lambda^{|I|} \right).
$$  
 For any $v \in V(G_k)$, the graph $H_{k,t}[Y_v]$ is a disjoint union of $t$ cliques of size $2^{2k}$. We obtain  
$$
\sum_{I\in \mathcal{I}(H_{k,t}[Y_v]) \setminus \{\emptyset\}} \lambda^{|I|} = (2^{2k} \lambda + 1)^t - 1
$$  
which gives the bound  
$$
\prod_{v\in J} \left( \sum_{I\in \mathcal{I}(H_{k,t}[Y_v])} \lambda^{|I|} \right) = ((2^{2k}\lambda+1)^t -1)^{s} 
$$  
Thus, we obtain  
\begin{align*}
    \pi(\Omega_S) &= \sum_{\substack{I \in \mathcal{I}(H_{k,t}) \\ |p(I)| = s} } \lambda^{|I|} \\
    &= \sum_{\substack{J \in \mathcal{I}(G_k) \\ |J|=s} } \sum_{\substack{I \in \mathcal{I}(H_{k,t}) \\ p(I) = J  }} \lambda^{|I|} \\[2em]
    &\leqslant 2^{2\lfloor c_1 k\rfloor +1} \alpha(G_k)^{\lfloor c_1k\rfloor} ((2^{2k}\lambda +1)^{t}-1)^s.
\end{align*}  
The final result follows by bounding $\alpha(G_k)$ by $2^{k+1}$.
\end{proof}

\begin{claim}
$\min(\pi(\Omega_I),\pi(\Omega_J)) \geqslant ((2^{2k}\lambda+1)^t -1)^{\alpha(G_k)}$.
\end{claim}

\begin{proof}
First, observe that $\Omega_I$ contains all independent sets $I\in \mathcal{I}(H_{k,t})$ such that $p(I) = I_k$. Using a similar argument as in the proof of Claim~\ref{claim:LowerBoundOmegaS}, we obtain  
$$
\pi(\Omega_I) \geqslant \sum_{\substack{I \in \mathcal{I}(H_{k,t}) \\ p(I) = I_k}} \lambda^{|I|} = ((2^{2k}\lambda+1)^t-1)^{\alpha(G_k)}.
$$

Next, we show that a similar bound holds for $\pi(\Omega_J)$. Specifically, we prove that $\Omega_J$ contains all independent sets $J\in \mathcal{I}(H_{k,t})$ such that $p(J) = J_k$.  

Assume for contradiction that there exists $J\in \Omega_I$ such that $p(J) = J_k$. Then, there must exist a path $W_0, W_1, \dots, W_p$ in $\mathcal{G}(H_{k,t})$ such that $W_i\in \Omega_I$ for all $0\leqslant i \leqslant p$ and $W_0=I_{k,t}$. We show that $p(W_0), p(W_1), \dots, p(W_p)$ forms a path from $I_k$ to $J_k$ in $\mathcal{G}(G_k)$.  Indeed, we have $p(W_0)= I_k$ and $p(W_p)=J_k$. Consider any step $0\leqslant i < p$ where $W_{i+1} = W_i\cup \{x\}$ for some vertex $x\in V(H_{k,t})$. Let $v\in V(G_k)$ be the unique vertex such that $x\in Y_v$. There are two cases:  
\begin{itemize}
\item If there exists a vertex $y\in W_i\cap Y_v$, then $p(W_{i+1})=p(W_i)$.  
\item If $W_i\cap Y_v = \emptyset$, then $p(W_{i+1}) = p(W_i)\cup \{v\}$.  
\end{itemize}
A similar argument holds when $W_{i+1} = W_i \setminus \{x\}$ for some vertex $x\in V(H_{k,t})$, ensuring that $|p(W_{i+1})\Delta p(W_i)|\leqslant 1$.  Since no index $0\leqslant i\leqslant p$ satisfies $|p(W_i)|=s$ (otherwise $W_i$ would belong to $\Omega_S$), this path from $I_k$ to $J_k$ in $\mathcal{G}(G_k)$ avoids all independent sets of size~$s$, contradicting Lemma~\ref{lemma:PropertiesGk}.(\ref{sublemma:SmallIS}).  Thus, we conclude that $\Omega_J$ contains all independent sets $J\in \mathcal{I}(H_{k,t})$ such that $p(J) = J_k$, implying  
$$
\pi(\Omega_J) \geqslant ((2^{2k}\lambda+1)^t -1)^{\alpha(G_k)}
$$
This completes the proof of the claim.
\end{proof}
Using Theorem~\ref{thm:LowerBoundConductance}, we obtain:  
\begin{align*}
     \tau_{H_{k,t},\lambda} &\geqslant \frac{\ln(2)}{4}\left( \frac{((2^{2k}\lambda+1)^t -1)^{\alpha(G_k)}}{2^{\lfloor c_1k \rfloor (k+3)+1}((2^{2k}\lambda +1)^{t}-1)^s}-2\right) \\[1em]
    &= \frac{\ln(2)}{4}\left( \frac{((2^{2k}\lambda+1)^t -1)^{\lfloor c_1k\rfloor}}{2^{\lfloor c_1k \rfloor(k+3)+1}}-2\right)
\end{align*}
Let $n = |V(H_{k,t})|$. We have  $n = t 2^{2k} (2^{k+2}-2) \leqslant t2^{3k+2}$, which implies 
$\log_2(n) \leqslant \log_2(t)+3k+2$. Rearranging, we obtain $k \geqslant \frac{1}{3}(\log_2(n)-\log_2(t) -2)$.

It is straightforward to observe that there exists a threshold $k_{t}$ such that for any $k\geqslant k_{t}$, the following inequalities hold:  
$$
\left\lbrace 
\begin{array}{l}
     k \geqslant \frac{\ln(2)}{4}\log_2(n),  \\[1em]
     ((2^{2k}\lambda+1)^t -1)^{\lfloor c_1k\rfloor} \geqslant 2^{\frac{4}{3}\lfloor c_1k\rfloor k t+3c_1k+3}\lambda^{\frac{2}{3}c_1kt}
\end{array}
\right.
$$
Thus, for any $k\geqslant k_{t}$, we obtain  
$$
\tau_{H_{k,t}, \lambda}\geqslant \lambda^{\frac{2}{3}c_1kt}2^{\lfloor c_1k\rfloor k\left( \frac{4}{3}t-1\right)} 
$$
and finally since $k\geqslant \frac{\log_2(n)}{4}$ and $\frac{4}{3}t-1 \geqslant \frac{t}{3}$ whenever $t\geqslant 1$, there exists a small enough constant $c>0$ such that 
$$
\tau_{H_{k,t}, \lambda} \geqslant (\lambda n)^{ct\ln n}
$$
\end{proof}

\section{Conclusion}

We showed that the mixing time of the Glauber dynamics for the hard-core model is subexponential on geometric intersection graphs. Additionally, we established tight parameterized mixing time bounds for this Markov chain, enabling a more fine-grained understanding of mixing times beyond the classical polynomial-versus-exponential dichotomy. 

In particular, for geometric intersection graphs, our results demonstrate that the Glauber dynamics is as efficient as exact algorithms for computing the partition function and activation probabilities of each vertex, while being significantly more practical to implement. Moreover, our approach does not require knowledge of the geometric representation of the graph. 

 We recall two major open questions and propose potential future directions:

\begin{itemize}
    \item For any $d\geqslant 3$, does the Glauber dynamics mix in time $2^{\Omega(n^{1-\frac{1}{d}})}$ on the $d$-dimensional grid with $n$ vertices? Such a result would require generalizing the work of Randall~\cite{randall2006slow} to higher dimensions.
    \item For any $\lambda > 1$, can the mixing time $\tau_{G,\lambda}$ on $n$-vertex graphs be bounded by a function of the form $n^{o(\tw(G))}$, or even by $f(\tw(G)) \cdot n^{O(1)}$, where $f$ is a computable function?

    \item On planar graphs, is it possible to combine the subexponential and parameterized approaches to achieve such a mixing time? Indeed, all the lower bounds we presented involve large complete bipartite induced subgraphs, which do not exist in geometric intersection graphs.
    \item Does there exist a chordal graph $G$ such that $\tau_{G,\lambda}=\Omega(n^{\ln n})$ for a large enough constant~$\lambda$?
\end{itemize}

\subsubsection*{Acknowledgements} I would like to express my sincere gratitude to Rémi Watrigant for his guidance throughout this work. I also thank Julien Duron for insightful discussions on the subject, and Mark Jerrum for pointing out relevant references. I would also like to thank Thomas Begin, Anthony Busson, and Loïc Chassin de Kergommeaux for our discussions regarding the motivation behind this study. Finally, I greatly thank anonymous reviewers who pointed out the links with the work of Dyer, Greenhill and Müller, which improved significantly the relevance of this work.

\bibliographystyle{plain}
\bibliography{biblio}

\appendix

\section{Preliminaries on Markov chains}\label{appendix:Markov}

This section provides a brief overview of Markov chains and the canonical paths technique for analyzing mixing times. We use the notation of \cite{sinclair1992improved}, and we also refer to \cite{guruswami2016rapidly,levin2017markov,montenegro2006mathematical} for more details.

\paragraph*{Markov chains}

A finite Markov chain $\mathcal{M}$ is defined by a state space $\Omega$ and a transition matrix $P$ of dimension $|\Omega| \times |\Omega|$, where $P(x, y)$ represents the probability of transitioning from $x \in \Omega$ to $y \in \Omega$. A step of the Markov chain corresponds to a single such transition. A distribution $\pi$ on $\Omega$ is stationary if $\pi P = \pi$. In an \emph{ergodic} Markov chain, the stationary distribution $\pi$ is unique and represents the long-term behavior of the chain as $t \to \infty$. Given an initial distribution $\mu$, the distribution of the chain after $t$ steps is given by $\mu P^t$, and for an ergodic chain, we have $\mu P^t \to \pi$ as $t \to \infty$.

The \emph{mixing time} of a Markov chain is the number of steps required for the chain to become $\varepsilon$-close to its stationary distribution. Given two distributions $\mu$ and $\nu$ on $\Omega$, the \emph{total variation distance} $\Delta(\mu,\nu)$ is defined as:
$$
\Delta(\mu,\nu)= \frac{1}{2}\sum_{x\in \Omega} |\mu(x) - \nu(x)|
$$

Starting from a state $x \in \Omega$, let $\mu_x$ denote the probability distribution where $\mu_x(x) = 1$ and $\mu_x(y) = 0$ for all $y \in \Omega \setminus {x}$. For $\varepsilon \in (0,1)$, the mixing time $\tau_x(\varepsilon)$ is the smallest $t$ such that
$$
\Delta(\mu_x P^t, \pi) \leqslant \varepsilon
$$
The mixing time is closely linked to the spectral properties of the transition matrix. By standard linear algebra, whenever $\mathcal{M}$ is ergodic, $P$ has $N:=|\Omega|$ real eigenvalues satisfying $1=\lambda_0 \geqslant \lambda_1 \geqslant ... \geqslant \lambda_{N-1} \geqslant -1$. Let $\lambda_{\max}(\mathcal{M}) = \max(\lambda_1, |\lambda_{N-1}|)$.

\begin{proposition}[\cite{sinclair1992improved}]\label{prop:BoundMixingTime}
Let $\varepsilon \in (0,1)$. The mixing time satisfies the following bounds:
\begin{enumerate}
    \item For any state $x \in \Omega$, $\tau_x(\varepsilon) \leqslant (1-\lambda_{\max})^{-1}\left( \ln\left( \frac{1}{\pi(x)}\right) +\ln\left( \frac{1}{\varepsilon}\right)\right)$
    \item $\max_{x\in \Omega} \tau_{x}(\varepsilon) \geqslant \frac{1}{2} \lambda_{\max} (1-\lambda_{\max})^{-1} \ln\left(\frac{1}{2\varepsilon}\right)$
\end{enumerate}
\end{proposition}

In the following, we assume that $\lambda_{\max} = \lambda_1$. While this is not always true, it holds in particular for Glauber dynamics and, more generally, for any heat-bath Markov chain~\cite{dyer2014structure}.

\paragraph*{Conductance}

In combinatorics, it has long been known that the second-largest eigenvalue of the adjacency matrix of a graph is closely related to its \emph{expansion}, a structural property of the graph. The conductance of $\mathcal{M}$ is defined as the edge expansion of the underlying transition graph:
\begin{equation}\label{eq:DefConductanceMarkov}
\Phi(\mathcal{M}) = \min_{\substack{S\subseteq \Omega \ 0\leqslant \pi(S)\leqslant 1/2}} \frac{Q(S,\overline{S})}{\pi(S)},
\end{equation}
where $\pi(S) = \sum_{x\in \Omega} \pi(x)$ and $Q(S,\overline{S}) = \sum_{x\in S, y\notin S} \pi(x)P(x,y)$. Note that $Q(S,\overline{S}) = Q(\overline{S},S)$ since for any $(x,y) \in \Omega^2$, we have $\pi(x)P(x,y) = \pi(y)P(y,x)$.

\begin{theorem}[\cite{sinclair1992improved}]\label{thm:ConductanceSpectrum}
Given an ergodic Markov chain $\mathcal{M}$ with second eigenvalue $\lambda_1$ and conductance $\Phi$, the following inequality holds :
$$
1-2\Phi \leqslant \lambda_1 \leqslant 1-\frac{\Phi^2}{2}
$$
\end{theorem}

The conductance is the main tool to prove lower bound on the mixing time of a Markov chain. One of the simplest methods to bound the conductance, introduced by Randall~\cite{randall2006slow}, involves partitioning the state space $\Omega$ into three subsets $(\Omega_S, \Omega_1, \Omega_2)$ such that $\Omega_S$ \emph{separates} $\Omega_1$ and $\Omega_2$. This means that for any pair of states $(x, y) \in \Omega_1 \times \Omega_2$ and any path $\gamma_{x,y} = (x = z_0, z_1, \ldots, z_\ell = y)$ satisfying $P(z_i, z_{i+1}) > 0$ for all $i$, there exists at least one intermediate state $z_i$ such that $z_i \in \Omega_S$. For the sake of completeness, we now formally state and prove the following theorem.

\begin{theorem}
Let $\mathcal{M}$ be an ergodic Markov chain such that $\lambda_{\max}(\mathcal{M})=\lambda_1(\mathcal{M})$, and let $(\Omega_S, \Omega_1,\Omega_2)$ be a partition of the state space $\Omega$ of $\mathcal{M}$ such that $\Omega_S$ separates $\Omega_1$ and $\Omega_2$. Thus, for any $\varepsilon\in (0,1)$, 
$$
\max_{x\in \Omega}\tau_x(\varepsilon) \geqslant \frac{1}{4}\left(\frac{\pi(\Omega_1)}{\pi(\Omega_S)}-2\right)\ln\left(\frac{1}{2\varepsilon}\right)
$$
\end{theorem}
\begin{proof}
Without loss of generality, we assume that $\pi(\Omega_1)\leqslant \pi(\Omega_2)$, and thus $\pi(\Omega_1)\leqslant 1/2$. By Equation~\ref{eq:DefConductanceMarkov} with $S = \Omega_1$ or $S=\Omega_2$ such that $\pi(S)$ is minimum,  we obtain that 
    $$
    \Phi \leqslant \frac{Q(\Omega_1, \Omega_S\cup \Omega_2)}{\pi(\Omega_1)} = \frac{Q(\Omega_1,\Omega_S)}{\pi(\Omega_1)}
    $$
    where the last equality holds since for any $(x,y)\in \Omega_1\times \Omega_2$, $P(x,y)=0$. Finally, we have 
    $$
    Q(\Omega_S,\Omega_1) =\sum_{(x,y)\in \Omega_S\times \Omega_1} \pi(x) P(x,y) = \sum_{x\in \Omega_S} \pi(x) \underbrace{\left( \sum_{y\in \Omega_1}P(x,y)\right)}_{\leqslant 1} \leqslant \pi(\Omega_S)
    $$
    It follows that $\Phi \leqslant \frac{\pi(\Omega_S)}{\pi(\Omega_1)}$. By Theorem~\ref{thm:ConductanceSpectrum}, we have 
    $$
    \lambda_1 \geqslant 1-2\Phi \geqslant 1-2 \frac{\pi(\Omega_S)}{\pi(\Omega_1)}
    $$
    Using Proposition~\ref{prop:BoundMixingTime}, we obtain 
    $$
        \max_{x\in \Omega}\tau_x(\varepsilon) \geqslant \frac{1}{2}\frac{\lambda_1}{1-\lambda_1}\ln\left(\frac{1}{2\varepsilon}\right)
        \geqslant \frac{1}{4}\left(\frac{\pi(\Omega_1)}{\pi(\Omega_S)}-2\right)\ln\left(\frac{1}{2\varepsilon}\right)
    $$
\end{proof}
\paragraph*{Canonical paths}

The \emph{canonical paths} technique provides a powerful method for upper bounding the mixing time. For each pair of states $x, y \in \Omega$, fix a specific path $\gamma_{x,y} = (x = z_0, z_1, \ldots, z_\ell = y)$ such that $(z_i, z_{i+1})$ are adjacent states, i.e., $P(z_i, z_{i+1}) > 0$. The path $\gamma_{x,y}$ is called the canonical path from $x$ to $y$. Let $\Gamma := \{\gamma_{x,y} \mid x, y \in \Omega\}$ be the set of all fixed canonical paths. The \emph{congestion} through a transition $e = (u, v)$, where $P(u, v) > 0$, is defined as
\begin{center}
\begin{equation}\label{CongestionMarkov}
\rho(\Gamma, e) := \frac{1}{\pi(u)P(u, v)} \sum_{\substack{x, y \in \Omega \\ \gamma_{x,y} \text{ uses } e}} \pi(x) \pi(y) |\gamma_{x,y}|
\end{equation}
\end{center}
where $|\gamma_{x,y}|$ is the length of the path $\gamma_{x,y}$. The overall congestion of the paths $\Gamma$ is then
$$
\rho(\Gamma) := \max_{e = (u,v) : P(u,v) > 0} \rho(\Gamma, e).
$$
The heart of this method is the following, 
\begin{theorem}[\cite{sinclair1992improved}]\label{thm:BoundSpectrumPaths}
Given an ergodic Markov chain $\mathcal{M}$ with second eigenvalue $\lambda_1$ and any choice of canonical paths $\Gamma$, the following inequality holds :
$$
\lambda_1 \leqslant 1-\frac{1}{\rho(\Gamma)}
$$
\end{theorem}

Combining Proposition~\ref{prop:BoundMixingTime} and Theorem~\ref{thm:BoundSpectrumPaths}, we obtain directly the following.

\begin{theorem}
Let $\mathcal{M}$ be an ergodic Markov chain such that $\lambda_{\max}(\mathcal{M})=\lambda_1(\mathcal{M})$, and let $\Gamma$ be a set of canonical paths. For any $\varepsilon \in (0,1)$ and $x\in \Omega$, 
$$
\tau_x(\varepsilon) \leqslant \rho(\Gamma) \left(\ln\left(\frac{1}{\pi(x)}\right) + \ln\left(\frac{1}{\varepsilon}\right)\right)
$$
\end{theorem}

\end{document}